\documentclass[12pt]{article}
\usepackage{amsmath}
\usepackage{amssymb}
\usepackage{amsthm}
\usepackage{mathtools}
\usepackage{graphicx}
\usepackage{xcolor}
\usepackage[numbers]{natbib}
\usepackage{url}
\usepackage{algorithm}
\usepackage{subcaption}
\usepackage{authblk}
\usepackage{booktabs}
\usepackage{bm}
\usepackage{rotating}
\usepackage[noend]{algpseudocode}
\usepackage[textwidth=6in,bottom=1in,top=1in,left=1in,right=1in]{geometry}
\usepackage[font=small,labelfont=bf]{caption}
\usepackage[toc,page]{appendix}
\usepackage{hyperref}

\newtheorem{theorem}{Theorem}
\newtheorem{lemma}{Lemma}

\newcommand{\data}{y}
\newcommand{\trueproc}{\xi}
\newcommand{\disc}{\delta_\theta}
\newcommand{\measerr}{\epsilon}
\newcommand{\model}{z(t_i,\theta)}
\newcommand{\param}{\theta}
\newcommand{\state}{z}

\title{A Sensitivity-based Framework for Calibrating Coupled Ordinary Differential Equations under Model Discrepancy}

\author[1,$^\dagger$]{Mitchel J. Colebank}
\author[2,$^\dagger$]{William Consagra}
\affil[1]{Department of Mathematics, University of South Carolina, Columbia, SC} 
\affil[2]{Department of Statistics, University of South Carolina, Columbia, SC}
\affil[$^\dagger$]{Authors Contributed Equally}

\begin{document}
\maketitle

\begin{abstract}
Coupled systems of ordinary differential equations (ODEs) are widely used to model complex physical and biological processes. In these applications, ODE model parameters must be calibrated to field data to enable prediction and parameter inference. In practice, the governing ODEs are an imperfect approximation to the true system, leading to non-negligible model discrepancy that must be incorporated to avoid biased parameter estimates. However, it is well known that highly flexible discrepancy models can confound discrepancy and simulator parameters, leading to poor identifiability. As a result, practitioners often must impose problem-specific prior constraints to adequately regularize the discrepancy, which can be challenging and cumbersome. In this work, we propose a novel calibration framework for coupled, multi-output ODE systems that enforces automatic, model-structural constraints on the discrepancy, improving identifiability without requiring application-specific discrepancy priors. Forward-model gradients are computed via sensitivity equations, solved jointly with the ODE system as a coupled initial value problem. Posterior inference is performed using an adaptive Metropolis-within-Gibbs sampler tailored to the resulting constrained posterior. We demonstrate the approach on three model systems: a mass-spring oscillator, a model of infectious disease spread, and a neuron firing model. We show that our method leads to improved parameter inference and predictive accuracy compared to standard calibration approaches.
\end{abstract}

\section{Introduction}
Model calibration is a necessary component in how simulation models are used for real-world decision making \cite{NatAcadDT,Smith2024}. The integration of data with a computer model, often called the \textit{inverse problem}, leads to a data-specific simulator that can be used for understanding an underlying process. In the case of mechanistic models, such as systems of differential equations, the model equations themselves embed hypotheses about the processes observed, and are almost always parameterized by interpretable variables that may or may not be measured \cite{Colebank2024G,Smith2024}. These simulation platforms form an important foundation for emerging  \textit{digital twin} technologies \cite{NatAcadDT}, as mechanistic models provide a way of forecasting the behavior of the system while being continually informed or updated by data. 
\par 
\par 
%Uncertainty quantification (UQ) comes in many mathematical forms, including quantifying model sensitivity based on prior beliefs in the parameters \citep{Alexanderian2021,Colebank2024,Kimpton2025}, identifying uncertainty in the model parameters \cite{consagra2024neural,Paun2020}, and propagating uncertainty through the simulator to quantify variability in output space \cite{Paun2020,Smith2024}. 
Accounting for \textit{model discrepancy}, defined as the systematic difference between the true physical process and its simulator, is critical when calibrating differential-equation-based models to ensure reliable inference, prediction, and uncertainty quantification (UQ) \citep{brynjarsdottir2014}.
%which is paramount for the transition towards digital twin enabled decision making \citep{NatAcadDT}. 
To address this important aspect of model calibration, in this work we investigate both parameter and predictive inference in the calibration of multi-output ordinary differential equation (ODE) models from noisy data under model discrepancy. Specifically, we consider data arising from a statistical model of the form
\begin{equation}\label{eqn:noise_inv_problem}
\begin{aligned}
        \data_i &= \trueproc(t_i) + \measerr_i \\
        &= Az(t_i,\theta) + \disc(t_i) + \measerr_i; \quad i=1,...,n.\\   
\end{aligned}
\end{equation}
In eq. \eqref{eqn:noise_inv_problem}, $A:\mathbb{R}^{J\times K}$ is a linear measurement operator,  $t_i\in\Omega\subset\mathbb{R}$ are observation times, $\param\in \Theta\subset\mathbb{R}^{p}$ is the unknown true calibration parameter, $\trueproc(\cdot)$ is the unknown true physical reality, and $\model:\Omega\mapsto \mathbb{R}^{K}$ is the forward model at some $t_i$. The term $\disc(t_i) = \trueproc(t) - \model$ is the model discrepancy and $\measerr_i$ is a mean-zero measurement error process. The matrix $A:\mathbb{R}^{K}\mapsto \mathbb{R}^{J}$ maps the full state vector $z(t_i,\theta)\in \mathbb{R}^{K}$ to the observable state space $\mathbb{R}^{J}$.  Given noisy observations $\mathcal{D} := \{(t_i,\data_i)\}_{i=1}^n$ from the physical process, the primary goals are (i) to infer $\param$ and quantify its uncertainty and (ii) to improve predictions of the physical process while quantifying their uncertainty.
\par 
The contribution of this work is a novel calibration framework that extends the popular class of orthogonal Gaussian process discrepancy priors \citep{Plumlee03072017} from univariate calibration to the calibration of mulit-output ODE simulators using multivariate observations. They key to our approach is the development of a multivariate discrepancy prior which is orthogonal to the forward model sensitives. We develop a complete analytical framework and scalable computational inference algorithm based on model sensitivity equations for our \textit{multivariate orthogonal sensitivity} (MOS) approach and apply it to several model types. We show that the MOS approach performs as well or better than the standard approach in both estimation and UQ.

The outline of the article is as follows. In Section \ref{sec:Background} we provide a brief overview of the general problem of calibration when model discrepancy is present, along with a review of calibration and sensitivity systems for ODE systems. Section \ref{sec:InferenceMethods} provides the mathematical preliminaries and derivation of the multivariate orthogonal discrepancy prior, as well as a computational inference algorithm for sampling from the relevant Bayesian posteriors. Section \ref{sec:NumericalExperiments} provides numerical details for our case studies, as well as three prototypical ODE system models: a mass-spring-damper system, the Fitzhu-Nagumo model of neuron activity, and a ``susceptible-infected-recovered-vaccinated'' (SIRV) model. Here, we investigate how our proposed MOS method compares to classical multivariate calibration approaches in both parameter and prediction metrics. We provide a discussion and conclusion in Section \ref{sec:Discussion}.

\section{Background}\label{sec:Background}

\subsection{Calibration with Model Discrepancy}  
\par 
Eq. \eqref{eqn:noise_inv_problem} is a special case of the calibration problem under model discrepancy. In the univariate output ($J=1$) setting, parameter inference and prediction for this model have received extensive treatment across the Bayesian \citep{brynjarsdottir2014,Chang2014,kennedy2001,Plumlee03072017,Xie2021}, frequentist \citep{Tuo2015,Wong2016}, and machine learning \citep{bhatnagar2022,huang2023learning,schmitt2024detectingmodelmisspecificationamortized,ward2022robust} literature. The influential approach by Kennedy and O'Hagan (KOH) \citep{kennedy2001} propose the use of a flexible Gaussian process prior on $\delta$. It is well known that this approach can suffer from identifiability issues between discrepancy and calibration parameters unless additional prior information is used to further constrain the prior on $\delta$ \citep{brynjarsdottir2014}.
%The influential approach in \cite{kennedy2001} propose an unconstrained GP prior on $\delta$. While leading to good predictions, $\delta$ can accommodate when the data $y$ is ``outside the range'' of the forward model, and this flexibility can make inference problematic \citep{brynjarsdottir2014}. 
A variety of remedies to this identifiability issue have been proposed \citep{Baker2022,sung2024review}. $L^2$-calibration \citep{Tuo2015} addresses the issue by redefining the true $\theta$ as an $L^2$ minimizer between the physical process and the forward model, 
%in our set-up,
%\begin{equation}\label{eqn:optimal_l2_param}
%    \theta = \min_{\eta\in\Theta}\|\xi(\cdot)-Az(\cdot,\eta)\|_{L^{2}(\Omega)}^2,
%\end{equation}
and propose an asymptotically consistent frequentist calibration procedure. Bayesian approaches constrain the discrepancy prior: \citep{Gu2018} encourages identifiability by imposing a prior to promote smaller (in $\|\cdot\|_{L^2}$ norm) discrepancies, while \citep{Plumlee03072017} use the first order conditions implied by $L^2$ optimality to linearly constrain a Gaussian process prior for $\delta_{\theta}$.
\par 
The $J\ge 2$ multivariate setting has garnered less attention. A common approach extends KOH by using independent multivariate Gaussian process priors \citep{higdon2008,PAULO20123959}.
%spitieris2023bayesian
However, this unconstrained approach can inherit similar or worse identifiability issues as outlined above in the univariate case. Constrained priors for $\delta_{\theta}$ in the multivariate setting are seemingly limited, though a recent work \citep{chakraborty2024orthogonal} proposes a framework based on posterior projections \citep{astfalck2024posteriorprojectioninferenceconstrained}, while \citep{white2023discrepancy} use PCA in the output space to map into a low-rank subspace, in which a simplified bi-linear model, motivated by identifiability considerations, is proposed for the discrepancy in the reduced space. A very recent study by Barnett and Chan \citep{Barnett2026} extended the $L^2$ orthogonality criteria from \citep{Plumlee03072017} to the multioutput setting. This last example considered discrepancy between Gaussian process surrogates and the true simulation model, but did not extend this to handing discrepancy between simulators and observations.

\subsection{Calibration for Multivariate ODE Systems}\label{ssec:calib_ode_systems}
We consider the case in which the state variables $z(t_i,\theta)$ in eq. \eqref{eqn:noise_inv_problem} evolve according to a system of ODEs. Systems of ODEs are a common modeling tool to describe the evolution of physical and biophysical processes \citep{Colebank2024G}. The broad definition of the system is given by
\begin{align}
\frac{d}{dt}z(t,\theta) = G(t,z(t,\theta),\theta), \quad z(0) = z_0,
\end{align}
with $z(t,\theta):\Omega\times \mathbb{R}^{p} \rightarrow \mathbb{R}^{K}$ being the state variables which evolve in time according to some evolution equation $G(t,z(t,\theta),\theta)$ from initial state $z_0$, and $\Omega\subset\mathbb{R}$ is the temporal domain. In the context of ODEs, the calibration problem in eq. \eqref{eqn:noise_inv_problem} now defines $A\model$ to be a mapping from the full states of the ODE system to the observations. This requires an additional level of care, as the solutions corresponding to observations are now approximations provided by numerical integration (e.g., using Runge-Kutta methods) at discrete time points for a possibly smaller subset of state variables (e.g., if $A$ is sparse). 

The most naive approach to solving these inverse problems is to (i) choose a current value for $\param$, (ii) numerically integrate the ODE system between $p+1$ and $2p+1$ times for a gradient approximation, and (iii) calculate a search direction by evaluating the residual between the observations and the subset of state solutions. This can become computationally expensive if either the ODE system is stiff, requiring small time-stepping, or if the dimensionality of the model is large, requiring repeated numerical integration in order to calculate the $p$ dimensional gradient of the objective function. 

Multiple prior studies across scientific domains have considered inverse problems governed by ODE systems. For instance, ODE models of temperature in air-cycle machines can be used for monitoring heat transfer in applications relevant to aerospace engineering \citep{Neal2019}. Epidemiological applications are especially well suited for using ODE models. Numerous applications have looked at simulating the transmission of disease from those who are susceptible to an infected and recovered state, often called the SIR model \citep{Tavener2011,Wang2026}. In neuroscience, models of neuronal excitation can be described through coupled ODE systems, where ion channels on a neuron can lead to excitatory and inhibitory interactions \citep{Doruk2019,Jensen2012}. While high-fidelity data on neuronal firing in rare, many groups use such ODE models to simulate whole-brain electrical activity as measured through imaging modalities \citep{consagra2024neural}. Finally, the cardiovascular field has extensively used ODE models with available clinical data for calibration \citep{Colebank2024,Lei2020,Shuttleworth2024}. These models include electrophysiological dynamics at the cell-level as well as compartment, organ-level models for simulating blood flow and overall cardiovascular function.
% MJC add some stuff here for papers that do use discrepancy versus ODE models

\subsection{Sensitivity Systems}
An alternative approach to approximating the gradient is to derive the so-called ``sensitivity equations" to provide the Jacobian of the system analytically \citep{Zi2011}. Consider the general form of the sensitivity of state $\state_j(t,\theta)$ with respect to parameter $\param_k$:
\begin{align}
    \mathcal{S}_{jk}(t) = \frac{\partial \state_j(t)}{\partial \param_k}.
\end{align}
Note that the sensitivities evolve in time with the dynamics of $z(t,\theta)$, and the total derivative of the differentiation equation with respect to $\param$ can be written as
\begin{align}
    \frac{d}{d\param}\frac{dz}{dt} = \frac{d}{dt}\frac{\partial z}{\partial \theta} &= \frac{d\mathcal{S}}{dt} = \frac{\partial}{\partial \theta}G(t,z(t,\theta),\theta) \\
    \implies \frac{d\mathcal{S}}{dt} &= \frac{\partial G}{\partial z} \frac{\partial z}{\partial \theta} + \frac{\partial G}{\partial \theta} =  \mathcal{G}_z\mathcal{S} + \mathcal{G}_\theta
    \end{align}
where $\partial G /\partial z  = \mathcal{G}_z \in \mathbb{R}^{(p \ \cdot K )\times (p \ \cdot K )}$ is a block diagonal system with the state-Jacobian, $\mathcal{J}\in\mathbb{R}^{p\times K}$, of the system on the diagonal, $\mathcal{S}$ is the sensitivity of the states with respect to the parameters, and $\partial G / \partial \theta = \mathcal{G}_\theta \in \mathbb{R}^{p \ \cdot K}$ is the parameter-Jacobian, with derivatives of the right hand side with respect to $\param$. We note that typically the sensitivity is written in matrix form with $\mathcal{S}\in\mathbb{R}^{K\times p}, \mathcal{G}_z \in \mathbb{R}^{K\times K},$ and $\mathcal{G}_\theta\in\mathbb{R}^{K\times P}$, but here we define everything in vectorized form for a cleaner translation to the available codes.

If the initial conditions of the ODE system are known, as will be assumed here, then $\mathcal{S}(0)=0$. Then, we can write the sensitivity in combination with the original ODE system as a coupled initial value problem:
\begin{align}
\frac{dz}{dt} &= G(t,z(t,\theta),\theta),   \ \ \ z(0) = z_0 \label{eq:states} \\
     \frac{d\mathcal{S}}{dt} &= \mathcal{G}_z \mathcal{S} + \mathcal{G}_\theta, \ \ \ \ \ \ \ \ \mathcal{S}(0) = 0 \label{eq:senses}
\end{align}
on the time interval $\Omega$.

Sensitivity systems have been the focus of some modeling studies, though typically these sensitivities are numerically approximated using finite differences. Some of the theoretical foundations for this approach were provided by seminal works by Leis et. al \citep{Leis1988} and Dickison et. al \citep{dickinson1976sensitivity} who derived such a sensitvity system for ordinary differential equation models. The review by Zi \citep{Zi2011} provides an excellent overview of how such sensitivity systems can be developed in biological modeling applications. Wu et al. \citep{Wu2008} provides examples of this in the context of biochemical processes, while work by Guay et. al \citep{Guay1995} apply this to chemical kinetic systems.
% MJC add some stuff here for papers that do use discrepancy versus ODE models

% Tools for differentiating ODE solution paths w.r.t. parameters, associated technical challenges of solving the forward problem, probably should mention emulation even if we don't treat it here ... use this to introduce relevant notation for the ODE-based forward model and sensitivity systems, e.g. what are the constraints on the ODE systems we study ($\theta$ cannot depend on $x$ for the current implementation, initial conditions can be rolled into $\theta$?, technical conditions on the class of ODE to make the differentials well defined, etc.

\section{Inference Methodology}\label{sec:InferenceMethods}

\subsection{Inferential Target and Observation Model}\label{ssec:ml2_model}
We assume that the sensitivity system defined in eq. \eqref{eq:senses} is well defined, smooth, and continuous \cite{Zi2011}.  We work in the inner product space induced by the direct sum of component-wise $L^2$ spaces $\mathcal{H}:=\bigtimes_{j=1}^J L^2(\Omega)$, with inner product defined as 
\begin{equation}\label{eqn:multivariate_l2_innter_prod}
\langle v, u \rangle_{\mathcal{H}} = \sum_{j=1}^J\alpha_j\langle v_j, u_j\rangle_{L^{2}(\Omega)},    
\end{equation}
where $\alpha_{j}\ge 0$ are optional weights. With loss of generality, for clarity of notation, from here on we will assume $\alpha_j=1$ $\forall j$. Clearly, $\langle \cdot, \cdot  \rangle_{\mathcal{H}}$ induces a norm, which we will denote as $\|\cdot \|_{\mathcal{H}}$. We define our inferential target using the $L^2$ calibration framework \citep{Tuo2015}, suitably extended to the multivariate setting using the loss induced by the norm $\|\cdot\|_{\mathcal{H}}^2$ to define the minimizer:\begin{equation}\label{eqn:optimal_calibration_param}
        \theta := \min_{\eta\in\Theta} \left\|\xi(\cdot) - Az(\cdot,\eta)\right\|_{\mathcal{H}}^2,
\end{equation}
where we assume $\Theta\subset\mathbb{R}^{p}$ is compact.
\par 
The observed data under model eq. \eqref{eqn:noise_inv_problem} with $J\ge 2$ is collected into the matrix $Y\in\mathbb{R}^{J\times n}$. We define the matrix of forward model and discrepancy evaluations at the fixed time points $t_1,...,t_n\in\Omega$ as $Z_{\theta},\Delta_{\theta}\in\mathbb{R}^{J\times n}$, respectively, with element-wise definitions $[Z_{\theta}]_{ji} = [Az]_j(x_i,\theta)$ and $[\Delta_{\theta}]_{ji} = \delta_j(x_i;\theta)$. We assume normality and homoscedasticity on the additive measurement errors with variance $\sigma^2$ and define the $J\times n$ matrix-valued random variable $E\sim \mathcal{MN}_{J\times n}(0, \sigma^2I_{J}, I_n)$, where 
$\mathcal{MN}$ denotes the matrix normal distribution. Thus, in matrix form \eqref{eqn:noise_inv_problem} becomes
\begin{equation}\label{eqn:multivariate_noise_inv_problem}
    Y = Z_{\theta} + \Delta_{\theta} + E 
\end{equation}
where, with slight abuse of notation, $Z_\theta$ incorporates the linear measurement operator. Taking a Bayesian approach, we are interested in sampling the marginal posterior density over the calibration parameter:
\begin{equation}\label{eqn:posterior_of_interest}
         \pi(\theta | Y) \propto \int \pi(Y|\Delta_{\theta},\theta)\pi(\Delta_{\theta}|\theta)\pi(\theta)d\Delta_{\theta}.
\end{equation}
Given the observation model eq. \eqref{eqn:multivariate_noise_inv_problem}, it follows that $\pi(Y|\Delta_{\theta},\theta)$ is a matrix normal density while $\pi(\theta)$ is the parameter prior density determined by problem specific information, e.g. uniform density over parameter bounds. In the following subsection, we propose a multivariate analog to the orthogonal discrepancy prior from \citep{Plumlee03072017} to construct the prior density $\pi(\Delta_{\theta}|\theta)$ based on the structural constraints imposed by Equation~eq. \eqref{eqn:optimal_calibration_param}.

\subsection{Multivariate Orthogonal Discrepancy Prior}\label{ssec:MOP}

We begin with the following lemma, which characterizes a collection of $p$ linear constraints for which the true multivariate discrepancy must obey under the definition of the optimal calibration in Equation~\ref{eqn:optimal_calibration_param}.

\begin{lemma}\label{lem:1st_order_orthogonality}
Denote $\mathcal{S}(t, \eta)\in\mathbb{R}^{K\times p}$ be the point-wise sensitivity of forward model $z(t,\eta)$, with element-wise definition $[\mathcal{S}(t,\eta)]_{jp}=\frac{\partial z_j(t,\eta)}{\partial\eta_p}$. Provided $\theta$ is in the interior of $\Theta$, the first order optimality conditions of $\theta$ from eq. \eqref{eqn:optimal_calibration_param} imply 
\begin{equation}\label{eqn:jac_subspace_oth}
    \int_{\Omega} \mathcal{S}(t, \theta)^{\intercal}A^{\intercal}\delta_{\theta}(t)dt = 0, \quad 0\in\mathbb{R}^p
\end{equation}
\end{lemma}

% \begin{lemma}\label{lem:1st_order_orthogonality}
% Denote $J_{f}(x, \eta)\in\mathbb{R}^{J\times p}$ be the point-wise Jacobian of forward model $f(x,\eta)$, with element-wise definition $[J_{f}(x,\eta)]_{jp}=\frac{\partial f_j(x,\eta)}{\partial\eta_p}$. Provided $\theta$ is in the interior of $\Theta$, the first order optimality conditions of $\theta$ from eq. \eqref{eqn:optimal_calibration_param} imply 
% \begin{equation}\label{eqn:jac_subspace_oth}
%     \int_{\Omega} J_{f}(x, \theta)^{\intercal}\delta_{\theta}(x)dx = 0
% \end{equation}
% where $0\in\mathbb{R}^p$ denotes the vector of zeros.
% \end{lemma}
Given these a-priori structural constraints, we desire a discrepancy prior such that the realization satisfy eq. \eqref{eqn:jac_subspace_oth} with probability one. The following result, which is a multivariate extension of Theorem 3 from \cite{Plumlee03072017}, provides a general construction for such a prior.
\begin{theorem}\label{thrm:constrained_multivariate_prior}
Let $C_{\phi}(t,t^{\prime})$ be a symmetric cross-variance function for a $J$-multivariate process with q-dimensional hyperparameters $\phi$. Let 
\begin{equation}\label{eqn:JacProjectionOperators}
\begin{aligned}
    h_{\theta}(t) &= \int_{\Omega}\mathcal{S}(t^{\prime},\theta)^{\intercal}A^{\intercal}C_{\phi}(t^{\prime},t)dt^{\prime}\in\mathbb{R}^{p\times J}\\
    H_{\theta} &= \int_{\Omega}\int_{\Omega}\mathcal{S}(t^{\prime},\theta)^{\intercal}A^{\intercal}C_{\phi}(t^{\prime},t)A\mathcal{S}(t,\theta)dtdt^{\prime}\in\mathbb{R}^{p\times p}
\end{aligned}
\end{equation}
and define the covariance function 
\begin{equation}\label{eqn:orthogonal_cross_covariance}
    C_{\theta,\phi}(t,t^{\prime}) = C_{\phi}(t,t^{\prime}) - h_{\theta}^{\intercal}(t)H_{\theta}^{-1} h_{\theta}(t^{\prime}),
\end{equation}
Then for $\delta_{\theta}\sim GP(0,C_{\theta,\phi})$, eq. \eqref{eqn:orthogonal_cross_covariance} holds with probability one. 
\end{theorem}
The quantities in eq. \eqref{eqn:JacProjectionOperators} require the evaluation of the Jacobian of the ODE system $z(\cdot, \theta)$ with respect to $\theta$, which can be done jointly with a forward pass of the model by solving the coupled sensitivity ODEs~\eqref{eq:states} and \eqref{eq:senses} to get $\mathcal{S}$.
% MJC: I don't think we need this section, do we?
% I agree 
% \subsection{Jacobian Calculation through Sensitivity System}
% The sensitivities $\mathcal{S}(t,z(t,\theta),\theta)$ represented the derivatives of the state variables $z(t,\theta)$. Hence, given some measurement operator $A$ as before, 

\subsection{Posterior Inference}\label{ssec:inf_alg}
Given the data and prior models outlined in Section~\ref{ssec:ml2_model} and Section~\ref{ssec:MOP}, respectively, our primary goal is to perform posterior inference for estimation and uncertainty quantification on $\theta$. We are also interested in inferring the model discrepancy. 
\par 
We begin by considering the calibration parameter posterior by calculating the marginalization integral eq. \eqref{eqn:posterior_of_interest} under the discrepancy prior outlined in Section~\ref{ssec:MOP}. Let $\text{vec}$ denote the vectorization operator and define the vectorized variables $\vec y := \text{vec}(Y)$, $\vec z_{\theta} := \text{vec}(Z_{\theta})$ and $\vec \delta_{\theta} := \text{vec}(\Delta_{\theta})$, where the stacking is performed columnwise, that is,  
$$
\text{vec}(\Delta_{\theta}) = \left(\vec\delta_1, \vec\delta_2, ..., \vec\delta_n\right)^{\intercal} \in\mathbb{R}^{Jn},
$$
where $\vec\delta_m = (\delta_1(t_m; \theta),...,\delta_J(t_m;\theta))\in\mathbb{R}^{J}$. Under the multivariate mean-zero GP prior with covariance function eq. \eqref{eqn:orthogonal_cross_covariance}, we have 
\begin{equation}\label{eqn:discretized_constrainted_kernel}
    \begin{aligned}
\text{Cov}\left[\vec \delta_{\theta}\right] &:= \begin{pmatrix}
\mathbb{E}\left[\vec\delta_1\vec\delta_1^{\intercal}\right] & \mathbb{E}\left[\vec\delta_1\vec\delta_2^{\intercal}\right] & \cdots & \mathbb{E}\left[\vec\delta_1\vec\delta_n^{\intercal}\right]  \\
\mathbb{E}\left[\vec\delta_2\vec\delta_1^{\intercal}\right] &   \mathbb{E}\left[\vec\delta_2\vec\delta_2^{\intercal}\right] &  \cdots & \vdots \\
\vdots & & \ddots \\
\mathbb{E}\left[\vec\delta_n\vec\delta_1^{\intercal}\right] & \cdots & & \mathbb{E}\left[\vec\delta_n\vec\delta_n^{\intercal}\right]
\end{pmatrix} 
 =\begin{pmatrix}
C_{\theta,\phi}(t_1,t_1) & C_{\theta,\phi}(t_1,t_2) & \cdots & C_{\theta,\phi}(t_1,t_n)  \\
C_{\theta,\phi}(t_2,t_1) &   C_{\theta,\phi}(t_2,t_2) &  \cdots & \vdots \\
\vdots & & \ddots \\
C_{\theta,\phi}(t_n,t_1) & \cdots & & C_{\theta,\phi}(t_n,t_n)
\end{pmatrix}. \\
\end{aligned} 
\end{equation}
Let $\vec t =(t_1,...,t_n)$ and 
with slight abuse of notation, we denote $\text{Cov}\left[\vec \delta_{\theta}\right]:= C_{\theta,\phi}(\vec  t, \vec  t)\in\mathbb{R}^{Jn\times J_n} $
and hence $\vec \delta_{\theta}\sim\mathcal{N}_{Jn}(0,C_{\theta,\phi}(\vec  t, \vec  t))$. \par 
From the linear model~eq. \eqref{eqn:multivariate_noise_inv_problem} we have that $Y|Z_{\theta}, \Delta_{\theta}$ is a matrix normal random variable: $\mathcal{MN}_{J\times n}(Z_{\theta} + \Delta_{\theta}, \sigma^2 I_{J}, I_n)$, and hence by properties of the matrix-normal we have the vectorized likelihood  
\begin{equation}\label{eqn:vectorized_lik}
    \pi(\vec y | z_{\theta},\delta_{\theta},\sigma^2) = \mathcal{N}_{Jn}(\vec y; \vec z_{\theta} + \vec \delta_{\theta} , \sigma^2I_{nJ}).
\end{equation}
Putting these identities together, we can write re-write the posterior of interest eq. \eqref{eqn:posterior_of_interest} in terms of the vectorized distributions as \begin{equation}\label{eqn:vectorized_parameter_posterior}
\begin{aligned}
     \pi(\theta | \vec y, \sigma^2, \phi) &\propto \int \mathcal{N}_{Jn}(\vec y; \vec z_{\theta} + \vec \delta_{\theta} , \sigma^2I_{nJ})\mathcal{N}_{Jn}(\delta_{\theta}; \vec 0, C_{\theta,\phi}(\vec t, \vec t)) \pi(\theta)d\vec \delta_{\theta} \\
     &\propto \mathcal{N}_{Jn}(\vec y;  \vec z_{\theta}, C_{\theta,\phi}(\vec t, \vec t) + \sigma^2I_{nJ})\pi(\theta) \\
     &\propto \frac{\exp\left(-\frac{1}{2}\left(\vec y - \vec z_{\theta}\right)^{\intercal}\left(  C_{\theta,\phi}(\vec t, \vec t) + \sigma^2I_{nJ}\right)^{-1} \left(\vec y - \vec z_{\theta}\right)\right)\pi(\theta)}{\det|C_{\theta,\phi}(\vec t, \vec t) + \sigma^2I_{nJ}|^{1/2}},
\end{aligned}
\end{equation}
where the second line follows from standard marginalization properties of Gaussian-Gaussian models.
\par 
We now calculate the conditional posterior of the discrepancy function at the observed  time points $\vec t$. By Bayes theorem, we have 
\begin{equation}\label{eqn:vectorized_discrepancy_conditional_posterior}
\begin{aligned}
\pi(\vec \delta_{\theta}|\vec y, \theta, \sigma^2, \phi) &\propto \pi(\vec y| \vec \delta_{\theta}, \theta)\pi(\delta_{\theta}|\theta,\phi) \\
&\propto \mathcal{N}_{Jn}(\vec y; \vec z_{\theta} + \vec \delta_{\theta} , \sigma^2I_{nJ})\mathcal{N}_{Jn}(\delta_{\theta}; \vec 0, C_{\theta,\phi}(\vec t, \vec t)) \\
&\propto \mathcal{N}_{Jn}\left(\vec \delta_{\theta}; \frac{1}{\sigma^{2}}\Gamma_{\theta}^{-1}\left(\vec y - \vec z_{\theta}\right), \Gamma_{\theta}^{-1}\right) \\
\end{aligned}
\end{equation}
where 
$$
\Gamma_{\theta} = \left[C_{\theta,\phi}(\vec t, \vec t)\right]^{-1} + \frac{1}{\sigma^2}I_{nJ},
$$
which follows from the Conjugate properties of Gaussian-Gaussian models. Then the conditional posterior of the discrepancy is given by 
\begin{equation}\label{eqn:vectorized_discrepancy_posterior}
    \pi(\vec \delta_{\theta}|\vec y, \sigma^2, \phi) \propto \int_{\Theta}\mathcal{N}_{Jn}\left(\vec \delta_{\theta}; \frac{1}{\sigma^2}\Gamma_{\theta}^{-1}\left(\vec y - \vec z_{\theta}\right), \Gamma_{\theta}^{-1}\right)\pi(\theta | \vec y, \sigma^2, \phi)d\theta.
\end{equation}
In general, eq. \eqref{eqn:vectorized_discrepancy_posterior} will not have a closed form. However, it can be sample from via Monte-Carlo integration, where draws from eq. \eqref{eqn:vectorized_parameter_posterior} are used to approximate the integral.
\par 
To construct a conditional posterior for the measurement error variance $\sigma^2$, we use an inverse-gamma prior $\pi(\sigma^2)=\operatorname{IG}(\sigma^2; \alpha_{\sigma^2},\beta_{\sigma^2})$. Then we have the closed form conditional posterior
\begin{equation}\label{eqn:conjugate_noise_var}
\begin{aligned}
    \pi(\sigma^2\mid \vec y,\theta,\vec\delta_{\theta})
    &\propto \mathcal{N}_{Jn}(\vec y; \vec z_{\theta} + \vec \delta_{\theta} , \sigma^2I_{nJ})\operatorname{IG}(\alpha_{\sigma^2},\beta_{\sigma^2}) \\
    &\propto \operatorname{IG}\left(\sigma^2;
    \alpha_{\sigma^2}+\frac{Jn}{2},
    \beta_{\sigma^2}+\frac{1}{2}\left|\vec y-\vec z_{\theta}-\vec\delta_{\theta}\right|_2^2\right),
\end{aligned}
\end{equation}
which follows from standard results in conjugacy. 

For the unconstrained discrepancy kernel hyper-parameters $\phi$, we assume independent uniform priors over problem-specific bounds, denoted collectively by $\pi(\phi)$. Extending eq.~\eqref{eqn:vectorized_parameter_posterior}, the joint conditional posterior of the calibration and kernel parameters is
\begin{equation}\label{eqn:joint_parameter_posterior}
\begin{aligned}
\pi(\theta,\phi\mid \vec y,\sigma^2) &\propto
\frac{
\exp\left[
-\frac{1}{2}
(\vec y-\vec z_{\theta})^{\intercal}
\left(
C_{\theta,\phi}(\vec t,\vec t)+\sigma^2I_{nJ}
\right)^{-1}
(\vec y-\vec z_{\theta})
\right]
\pi(\theta)\pi(\phi)
}{
\det\left|
C_{\theta,\phi}(\vec t,\vec t)+\sigma^2I_{nJ}
\right|^{1/2}
}.
\end{aligned}
\end{equation}

\begin{algorithm}[!h]
\caption{Adaptive Metropolis-within-Gibbs Posterior Inference}
\label{alg:inference_algo}

{\fontsize{11pt}{11pt}\selectfont
\begin{algorithmic}[1]
\State \textbf{Input} $\vec y$, $\vec t$, initial values
$(\theta^{(1)},\phi^{(1)},\sigma^{2(1)})$,
proposal covariance $V_0$, adaptation frequency $k_0$, number of iterations $M$.
\State \textbf{Output} Samples  
$\{\theta^{(m)},\phi^{(m)},\vec\delta^{(m)},\sigma^{2(m)}\}_{m=1}^{M}$
\State Set $p=\dim(\theta)$, $q=\dim(\phi)$, $c=2.38^2/(p+q)$, and
$R\gets\operatorname{chol}(cV_0+\epsilon I_{p+q})$
% \State Set $d=\dim(\psi)$, $c=2.38^2/(p+q)$, and
% $R\gets\operatorname{chol}(cV_0+\epsilon I_{p+q})$
\State Sample $\vec\delta^{(1)}$ from
eq.~\eqref{eqn:vectorized_discrepancy_conditional_posterior}, conditional on $\theta^{(1)},\phi^{(1)},\sigma^{2(1)}$.
\For{$m=1,\ldots,M-1$}

\State Propose
$([\theta^{\mathrm{prop}}],[\phi^{\mathrm{prop}}])=
([\theta^{(m)}],[\phi^{(m)}])
+R\psi$,
where $\psi\sim\mathcal{N}_{p+q}(0,I_{p+q})$
\State Get $\vec{z}(t,\theta^{\mathrm{prop}})$, $\mathcal{S}(\vec t, \theta^{\mathrm{prop}})$ by solving coupled ODEs \eqref{eq:states},
\eqref{eq:senses} 
\State Set $\theta^{(m+1)}=\theta^{\mathrm{prop}},\phi^{(m+1)} = \phi^{\mathrm{prop}}$ with probability
$$
\min\left\{
1,
\frac{
\pi(\theta^{\mathrm{prop}},\phi^{\mathrm{prop}}\mid
\vec y,\sigma^{2(m)})
}{
\pi(\theta^{(m)},\phi^{(m)}\mid
\vec y,\sigma^{2(m)})
}
\right\}
$$
\qquad Using eq.~\eqref{eqn:joint_parameter_posterior}.
Otherwise, set $\theta^{(m+1)}=\theta^{(m)},\phi^{(m+1)} = \phi^{(m)}$
\State Sample $\vec\delta^{(m+1)}$ from
eq.~\eqref{eqn:vectorized_discrepancy_conditional_posterior}
conditional on $\theta^{(m+1)}$, $\phi^{(m+1)}$, and $\sigma^{2(m)}$
\State Sample $\sigma^{2(m+1)}$ from
eq.~\eqref{eqn:conjugate_noise_var}
conditional on $\theta^{(m+1)}$ and $\vec\delta^{(m+1)}$
\If{$m$ is a multiple of $k_0$}
\State Set
$$
V_m\gets
\operatorname{Cov}\left(
([\theta^{(1)}],[\phi^{(1)}]),\ldots,([\theta^{(m+1)}],[\phi^{(m+1)}])
\right),
\qquad
R\gets\operatorname{chol}(cV_m+\epsilon I_d)
$$

\EndIf

\EndFor

\end{algorithmic}

}

\end{algorithm}

\par 
The posterior in eq.~\eqref{eqn:joint_parameter_posterior} does not generally admit a closed form. Therefore, $(\theta,\phi)$ are sampled jointly using the adaptive Metropolis algorithm \cite{Smith2024}. Algorithm~\ref{alg:inference_algo} provides the pseudocode for the complete posterior sampling procedure. At each iteration, $(\theta,\phi)$ are updated jointly using an adaptive Gaussian random-walk Metropolis step targeting the marginal posterior in eq.~\eqref{eqn:joint_parameter_posterior}. This involves solving the augmented ODE and sensitivity system for $\theta$, and then updating the discrepancy covariance using $\mathcal{S}$ and the current hyperparameter values $\phi$. Conditional on the accepted parameters, the discrepancy is sampled exactly from eq.~\eqref{eqn:vectorized_discrepancy_conditional_posterior}, and $\sigma^2$ is sampled using eq.~\eqref{eqn:conjugate_noise_var}. The proposal covariance is periodically adapted using the empirical covariance of the sampled $(\theta,\phi)$ values.
%\wc{Should include a sentence or two about line 7 + how the quadraditic form including $C_{\theta,\phi}$ is actually calculated.}

\section{Numerical Experiments}\label{sec:NumericalExperiments}

\subsection{Experiment Design}
\subsubsection{Data Generation}
We examine three underlying ODE models split into five sections. With slight abuse of notation, for each model type we use mathematical notation that is consistent with the modeling community. Thus, there may be overlap in the use of previously defined variables and each set of results should be interpreted within the context of the model domain. Each model is given specific upper and lower parameter bounds, denoted by UB and LB, respectively. For each model, we simulate observed data under observation model \eqref{eqn:noise_inv_problem} where model parameters $\theta$ are sampled from box-uniform priors over the problem specific bounds and model discrepancy is simulated from the MOS prior outlined in Theorem~\ref{thrm:constrained_multivariate_prior}. This ensures that the simulated model parameters remain the $L^2$ minimizer according to Equation~\ref{eqn:multivariate_l2_innter_prod} and are the proper inferential target. The kernel hyperparameters $\phi$ are generated from a problem specific log-normal distribution, where we simultaneously infer $\phi$ in posterior inference to simulate the realistic setting where the discrepancy hyperparameters are unknown. 
% The kernel hyperparameters $\gamma$ and $\ell$ are generated from a problem specific log-normal distribution and not supplied to the posterior inference Algorithm~\ref{alg:inference_algo} to simulate the realistic setting where discrepancy lengthscale and smoothness may be a-priori unknown. 
The measurement error is zero-mean Gaussian with a variance $\sigma_\varepsilon^2$ that is set based on the specific model system. Each experiment was repeated for $100$ Monte-Carlo replications and results combined for evaluation.

%For each dataset, we run a multistart optimization using 50 initializations of the simulator parameters using a Sobol' sequence [XXX]. We use the $L^2$ error as our objective function and select the minizer from these initializations as the ``true $L^2$ minimizer'' for our later analyses. 

\subsubsection{Implementation Details}\label{ssec:implementation_details}
All empirical results in this manuscript are generated using a custom Python implementation built on standard scientific computing libraries, including SciPy, NumPy, PyTorch, and GPyTorch, and made publicly available on Github (\url{https://github.com/mjcolebank/Discrepancy_USC}). Unless otherwise stated, default numerical and inferential settings are used throughout. The ODE system (with the state variables and sensitivites) are solved using SciPy's  \texttt{odeint} integration tool, which uses the LSODA algorithm \cite{LSODA}. Inputs ($t$) correspond to the fixed simulation grid from the ODE solver and are in their physical units. Uniform priors defined using problem specific LB and UB are used for the model parameters. For each proposed parameter value, the simulator is solved on a fixed, prescribed time interval to provide the state and sensitivity outputs. Solving sensitivities directly avoids extra simulations and reduces numerical errors for gradient approximation.  Sensitivity information is used only when building \eqref{eqn:orthogonal_cross_covariance}.  Posterior inference is performed using Algorithm~\ref{alg:inference_algo} for a single chain for 10,000 iterations. All inferences are conducted on the last 5,000 samples from the MCMC procedure, discarding the first  5,000 as burn-in. The initial proposal covariance is defined by the diagonal matrix $V^0=\textrm{diag}(\sigma_\theta)$, where $\sigma_\theta=0.01\left(\textrm{UB}-\textrm{LB}\right)$, and updates are made every $k_0=2,000$ iterations.

\subsubsection{Competing Methods}
We compare our method to a multivariate extension of the Kennedy O'Hagan approach \citep{kennedy2001} (KOH), where the unconstrained discrepancy prior is modeled using a mean-zero independent multivariate Gaussian process \cite{Paun2025}, and a radial basis function (RBF) for each output dimension, resulting in the Kronecker form 
\begin{equation}\label{eqn:discretized_unconstrainted_kernel}
    \mathcal{C}_{\phi}(\vec t,\vec t') = I_J \otimes C_{\phi}\left(\vec t, \vec t'\right) \in \mathbb{R}^{Jn \times Jn}, \ \ \ C_\phi(t_i,t_j)=\phi_1\exp\left( -\frac{(t_i - t_j)^2}{2\phi_2^2} \right).
\end{equation}
The kernel $C_{\phi}(t,t')$ is used in both the KOH and the MOS procedure, the latter of which uses $C_{\phi}(\vec t,\vec t')$ to calculate $C_{\theta,\phi}(\vec t,\vec t')$. Posterior inference using KOH was performed using the same Algorithm~\ref{alg:inference_algo}, with settings kept largely consistent with those outlined in Section~\ref{ssec:implementation_details}, except replacing any computations involving orthogonalized kernel~\eqref{eqn:discretized_constrainted_kernel} with the kernel~\eqref{eqn:discretized_unconstrainted_kernel} for KOH. 

\subsubsection{Evaluation Metrics}
We evaluate the estimation performance of the methods by computing the Monte-Carlo average of the root mean squared error (RMSE) and bias of the posterior mean. For UQ, we use the MCMC samples to form $95\%$ credible intervals and compute the empirical coverage and average interval length over the Monte-Carlo replications. 
\par 
We also evaluate inference on the true process $\xi$ by sampling the push-forward of the posterior parameter and discrepancy samples through the discrepancy adjusted forward model. Estimation performance is quantified by approximating the integrated squared error and integrated bias of the posterior predicted process. The quality of the uncertainty quantification was assessed by constructing $95\%$ point-wise credible intervals over the time grid and computing their empirical coverage and average length using the Monte-Carlo replications.

\subsection{Model 1: Mass-Spring-Damper Model}
For our first example, we consider a coupled mass-spring-damper system given by the second-order, linear matrix differential equation
\begin{align}
        {M} \ddot{x} + {C} \dot{x} + {K} {x} = {F}.   
    \end{align}
Here, we relate the acceleration $\ddot{x}$, velocity $\dot{x}$, and position $x$ of masses that are connected to each other through a spring-dashpot system. Given $n_s$ springs, the above equation has four sets of parameters: a set of $n_s$ mass terms which are described by the diagonal matrix of masses ${M}$, a set of $n_s$ dashpot elements (the dampers) parameterized by the bidiagonal matrix ${C}$, and $n_s$ stiffness parameters given by the bidiagonal matrix ${K}$. The vector ${F}$ includes the nonhomogenous component of the system, e.g., periodic forcing functions. We recast the system of $n_s$ second-order equations as a system of $2n_s$ first-order equations given by
\begin{equation}\label{eq:general_massspring}
    \frac{d}{dt}\begin{bmatrix}
        {x} \\ \dot{x}
    \end{bmatrix}
    =
    \begin{bmatrix}
        {0} & {I} \\
        -{M}^{-1}{K} & -{M}^{-1}{C} \ \ 
    \end{bmatrix} 
    \begin{bmatrix}
        {x} \\ {\dot{x}}
    \end{bmatrix}
    + \begin{bmatrix}
        0 \\ F
    \end{bmatrix}, 
\end{equation}
with known initial conditions $x(0)=x_0$ and $\dot{x}(0)=v_0$.

% which has homogeneous solution
% \begin{equation}
%     {z}(t) = e^{{A}t}{z}(0), \ \ {A}=
%     \begin{bmatrix}
%         {0} & {I} \\
%         -{M}^{-1}{K} & -{M}^{-1}{C} \ \ 
%     \end{bmatrix}
% \end{equation}

Inferring the elements of ${M}$,${C}$, ${K}$, and ${F}$ simultaneously leads to non-identifiability, since the mass terms are not uniquely determined in this system \cite{Smith2024}. Thus, we fix the masses ${M}$ and treat them as known quantities. We also leave the forcing terms ${F}$ fixed at their known value, and focus on inferring the parameters that make up the matrices ${C}$ and ${K}$ representing the dashpot and stiffness elements.

In general, the sensitivity system of eq.~\eqref{eq:general_massspring} can be derived by noting that, since the ODE system is linear in the states, the state-Jacobian is the evolution matrix
\begin{align}\label{eq:spring_Gz}
    \mathcal{J{}}   =
    \begin{bmatrix} 
    0 & I \\
    -M^{-1}K & -M^{-1}C
    \end{bmatrix}
\end{align}
which is true regardless of the number of springs. The term $\mathcal{G}_\theta$ also follows a pattern that is best seen on a case-by-case basis.

We begin with the case where there are $n_s=2$ masses, where the first mass is connected to a fixed surface and the second mass is free. The linear system is given by
\begin{align}
    \mathbf{M}=\begin{bmatrix}
m_1 & 0 \\[4pt]
0 & m_2\\[4pt]
\end{bmatrix}, \ \ \ 
\mathbf{K}=
\begin{bmatrix}
k_1+k_2 & -k_2 \\[4pt]
- k_2 & k_2\\[4pt]
\end{bmatrix}, \ \ \ 
\mathbf{C}=
\begin{bmatrix}
c_1+c_2 & -c_2 \\[4pt]
- c_2 & c_2\\[4pt]
\end{bmatrix}.
\end{align}
Our parameter vector of interest is $\theta = \{k_1,c_1,k_2,c_2\}$. We assume that the forcing function is applied to the second mass only, giving
\begin{equation}
    {F} = \begin{bmatrix}
0 \\[4pt]
F_0\cos(\omega_f \cdot t)\\[4pt]
\end{bmatrix}
\end{equation}
where $F_0$ is the forcing magnitude and $\omega_f$ is the forcing frequency. We fix the external force parameters for all simulations and data at $F_0=25$ and $\omega_f=5$. The corresponding sensitivity system then uses $\mathcal{G}_z=\textrm{diag}\left(\mathcal{J}\right)$ defined in eq.~\eqref{eq:spring_Gz}, while the term $\mathcal{G}_\theta\in\mathbb{R}^{16}$ has the non-trivial elements
\begin{align}
    &\mathcal{G}_{\theta_{3}} = -x_1/m_1,   &\mathcal{G}_{\theta_{7}} = -\dot{x}_1/m_1 & \\
    &\mathcal{G}_{\theta_{11}} = (-x_1+x_2)/m_1,  &\mathcal{G}_{\theta_{12}} = (x_1-x_2)/m_2, & \\
    &\mathcal{G}_{\theta_{15}} = (-\dot{x}_1+\dot{x}_2)/m_1,  &\mathcal{G}_{\theta_{16}} = (\dot{x}_1-\dot{x}_2)/m_2. 
\end{align}
We provide an example forward evaluation of the model and its sensitivities in Figure \ref{fig:Example_TwoSpring}. An averaged ranking of parameter influence, given by the sensitivity system across all 100 datasets, is provided in Figure \ref{fig:TwoSpring_Sensitivity} of the Appendix (Sec. \ref{sec:AppSens}).

\begin{figure}
    \centering
    \includegraphics[width=0.45\linewidth]{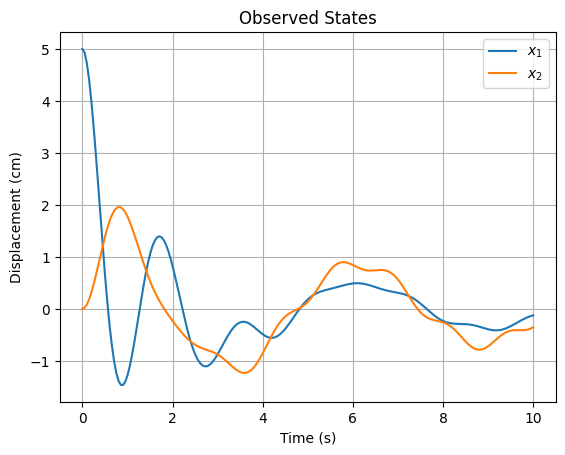}
    \includegraphics[width=0.45\linewidth]{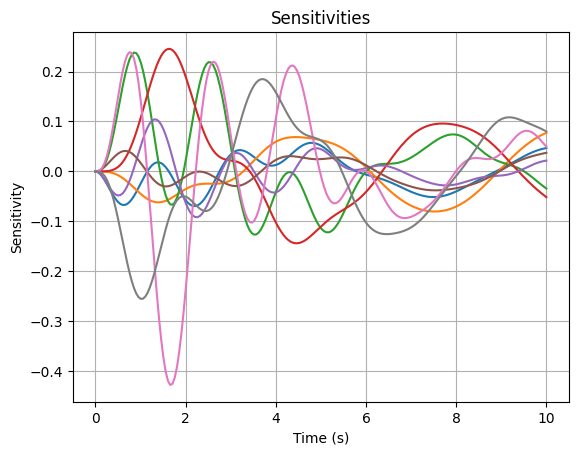}
    \caption{Example realization from the two-mass model. (Left): Example trajectory for the two mass displacements. (Right): The corresponding 8 sensitivities of each model state to the parameters of the system.}
    \label{fig:Example_TwoSpring}
\end{figure}

\subsubsection{Inference Results}
We consider the case where $A$ maps our model states to observations for the two mass displacements, $x_1$ and $x_2$, at 20 equally spaced time points on the interval $\Omega=[0,10]$. Here, our observation operator $A$ ignores the velocity states $\dot{x}_1$ and $\dot{x}_2$. Each dataset is generated by sampling parameters from a uniform distribution, with $k_i\sim \mathcal{U}(1,100)$, $i=1,2$, and $c_i\sim \mathcal{U}(0.1,10)$, $i=1,2$. The initial conditions for the model are $x_1(0)=5.0$ and $x_2=\dot{x}_1=\dot{x}_2=0.0$, with corresponding fixed, known masses of $m_1=5.0$ and $m_2=10.0$. We add independent zero-mean Gaussian noise to each dataset, with a noise variance $\sigma_\varepsilon^2=0.1$. For the model discrepancy, we generate a signal using the MOS framework with hyperparameters $\gamma \sim \textrm{LogNormal}\left(-2.0,1.0^2\right)$ and $\ell \sim \textrm{LogNormal}\left(-0.2,0.8^2\right)$. We use these priors for the GP to ensure the discrepancy magnitude is reasonable relative to the displacement magnitudes and a smaller length-scale to ensure sufficient discrepancy as the spring reaches steady state over the time interval $\Omega$.

Table \ref{tab:posterior_summary_metrics_twospring} provides the average bias, root mean square error (RMSE), 95\% coverage, and interval width over 100 Monte-Carlo experimental replications. In general, our method demonstrates superior parameter estimation compared to KOH, with smaller magnitude bias (except for $c_2$) and a smaller RMSE for all model parameters. For UQ, both methods achieve reasonable calibration, with empirical coverage for most parameters near the $95\%$ nominal level, though both exhibit slight conservatism. The mean interval widths for MOS are smaller than KOH for all model parameters, indicating less uncertainty in the marginal posterior. 
\par 
We also provide predictive error metrics in Table \ref{tab:predictive_metrics_twospring}, as well as one representative example of output uncertainty and marginal posteriors in Figure \ref{fig:TwoSpring_Output_UQ}. Similar patterns emerge, with MOS providing lower MSE and less conservative intervals, and empirical coverage closer to the nominal level than KOH. We observe that KOH does achieve a lower bias, though the bias for both methods is quite small.

\begin{table}[ht]
\centering
\footnotesize
\caption{Posterior summary metrics for the two-mass oscillator across 100 Monte-Carlo experimental replications. Bold values for the bias and RMSE are the best for the given parameter.}
\label{tab:posterior_summary_metrics_twospring}
\begin{tabular}{l|ccccc}
\hline
Method & Parameter & Mean Bias & Mean RMSE & Mean 95\% Coverage & Mean Int. Width \\
\hline
KOH & $k_1$ & 0.5533 & 3.1099 & 0.98 & 25.8976 \\
KOH & $c_1$ & 0.3660 & 0.9100 & 0.96 & 5.6902 \\
KOH & $k_2$ & -0.5922 & 3.7509 & 0.96 & 32.9859 \\
KOH & $c_2$ & \textbf{0.1553} & 0.6997 & 0.95 & 5.4625 \\
\hline
MOS & $k_1$ & \textbf{0.1767} & \textbf{2.5667} & 0.98 & 22.0292 \\
MOS & $c_1$ &\textbf{ 0.2961} & \textbf{0.6945} & 0.96 & 5.0494 \\
MOS & $k_2$ & \textbf{-0.4139} &\textbf{ 3.5067} & 0.96 & 30.1337 \\
MOS & $c_2$ & 0.1639 & \textbf{0.6294} & 0.95 & 5.0156 \\
\hline
\end{tabular}
\end{table}

\begin{table}[htbp]
  \centering
  \caption{Output uncertainty metrics between KOH and our MOS approach for the two-mass oscillator. Error metrics are first time-averaged and then averaged across 100 Monte-Carlo experimental replications. Bold values for the bias and RMSE are the best for the given parameter.}
  \label{tab:predictive_metrics_twospring}
  \begin{tabular}{l|cc}
    \hline\hline
    & \textbf{KOH} & \textbf{MOS} \\
    \cline{2-3}
    \textbf{Metric} & Mean & Mean \\
    \hline
    Time Integrated RMSE           & 0.3133 & \textbf{0.3083} \\
    Time Integrated Bias           & \textbf{-0.0174} & -0.0177 \\
    Time Integrated Coverage       & 0.9895 & 0.9788 \\
    Time Integrated Interval Width & 2.0935 & 1.8589 \\
    \hline\hline
  \end{tabular}
\end{table}

% \begin{table}[htbp]
%   \centering
%   \caption{Output uncertainty metrics between KOH and our MOS approach for the two-mass oscillator. Error metrics are first time-averaged and then averaged across 100 Monte-Carlo experimental replications. \wc{I'd drop medians...}}
%   \label{tab:predictive_metrics_twospring}
%   \begin{tabular}{l|cc|cc}
%     \hline\hline
%     & \multicolumn{2}{c}{\textbf{KOH}} & \multicolumn{2}{c}{\textbf{MOS}} \\
%     \cline{2-3} \cline{4-5}
%     \textbf{Metric} & Mean & Median & Mean & Median \\
%     \hline
%     Time Integrated RMSE           & 0.3133 & \textbf{0.2656} & \textbf{0.3083} & 0.2715 \\
%     Time Integrated Bias           & \textbf{-0.0174} & \textbf{-0.0100} & -0.0177 & -0.0102 \\
%     Time Integrated Coverage       & 0.9895 & 1.0000 & 0.9788 & 1.0000 \\
%     Time Integrated Interval Width & 2.0935 & 1.9280 & 1.8589 & 1.6365 \\
%     \hline\hline
% \end{tabular}
% \end{table}

\begin{figure}
    \centering
    \includegraphics[width=\linewidth]{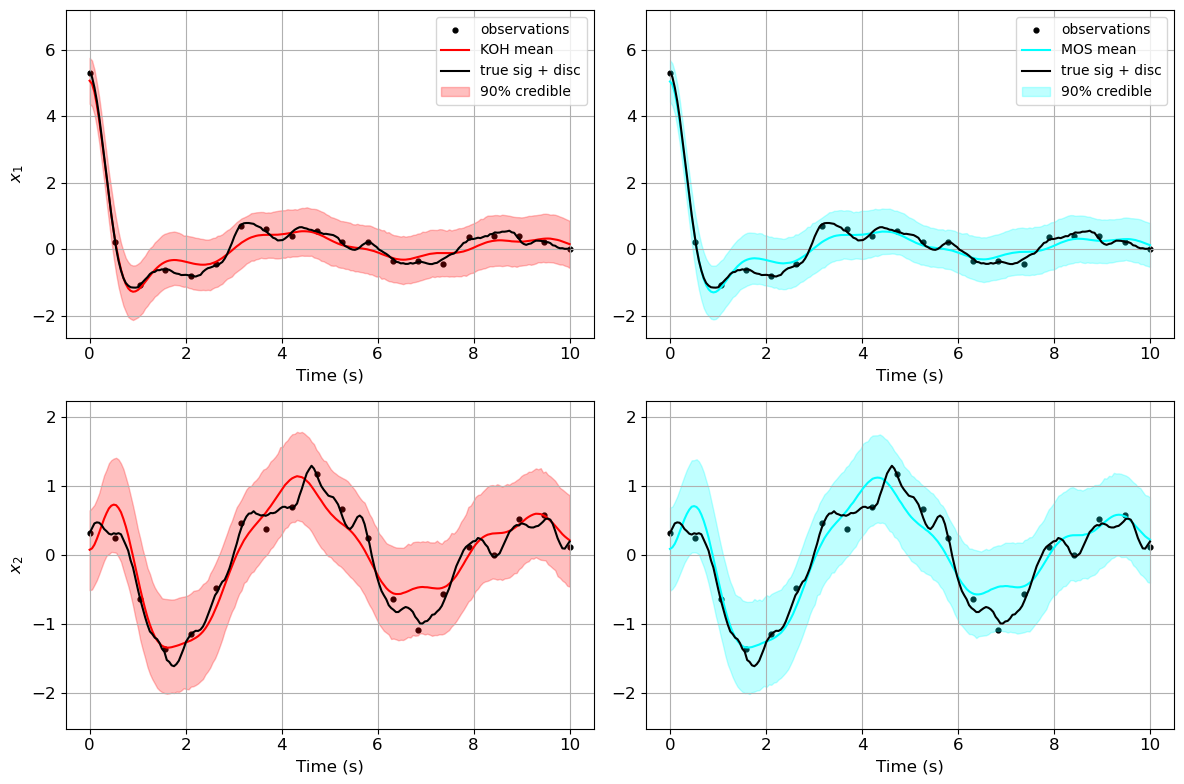}
    \caption{Output UQ for two-mass oscillator using one representative dataset. Posterior samples of the simulator and discrepancy are drawn from either the KOH (left column, red) or MOS (right column, cyan) approach.}
    \label{fig:TwoSpring_Output_UQ}
\end{figure}

\pagebreak

\subsection{Model 2: Four-Mass Oscillator}
We now consider a larger system of $n_s=4$ masses, with
\begin{align}
    \mathbf{M}=\begin{bmatrix}
m_1 & 0 & 0 & 0\\[4pt]
0 & m_2 & 0 & 0\\[4pt]
0 & 0 & m_3 & 0\\[4pt]
0 & 0 & 0 & m_4
\end{bmatrix}, \ \ \ 
\mathbf{K}=
\begin{bmatrix}
k_1+k_2 & -k_2 & 0 & 0\\[6pt]
- k_2 & k_2+k_3 & -k_3 & 0\\[6pt]
0 & -k_3 & k_3+k_4  & -k_4\\[6pt]
0 & 0 & -k_4 & k_4
\end{bmatrix}, \ \
\end{align}
\begin{align}
\mathbf{C}=
\begin{bmatrix}
c_1+c_2 & -c_2 & 0\\[6pt]
- c_2 & c_2+c_3 & -c_3\\[6pt]
0 & - c_3 & c_3+c_4 & -c_4\\[6pt]
0 & 0 & -c_4 & c_4
\end{bmatrix}.
\end{align}
We can recast the system of $4$ second-order equations as a system of $8$ first-order equations as we did before, where the forcing function ${F}$ is only applied to the last mass $m_4$ as before. The model parameters are now ${\theta}=\left\{k_1,c_1,k_2,c_2,k_3,c_3,k_4,c_4\right\}\in\mathbb{R}^8$. The sensitivity system requires $\mathcal{G}_\theta$ and $\mathcal{G}_z$ as before, where $\mathcal{G}_z$ is structured identically to eq.~\eqref{eq:spring_Gz}. The term $\mathcal{G}_\theta\in\mathbb{R}^{64}$ has the non-trivial elements
\begin{align}
    &\mathcal{G}_{\theta_{5}} = -x_1/m_1,   &\mathcal{G}_{\theta_{13}} = -\dot{x}_1/m_1 & \\
    &\mathcal{G}_{\theta_{21}} = (-x_1+x_2)/m_1,  &\mathcal{G}_{\theta_{22}} = (x_1-x_2)/m_2, & \\
    &\mathcal{G}_{\theta_{29}} = (-\dot{x}_1+\dot{x}_2)/m_1,  &\mathcal{G}_{\theta_{30}} = (\dot{x}_1-\dot{x}_2)/m_2, & \\
    &\mathcal{G}_{\theta_{38}} = (-x_2+x_3)/m_2,  &\mathcal{G}_{\theta_{39}} = (x_2-x_3)/m_3, & \\
    &\mathcal{G}_{\theta_{46}} = (-\dot{x}_2+\dot{x}_3)/m_2,  &\mathcal{G}_{\theta_{47}} = (\dot{x}_2-\dot{x}_3)/m_3, & \\
    &\mathcal{G}_{\theta_{55}} = (-x_3+x_4)/m_3,  &\mathcal{G}_{\theta_{56}} = (x_3-x_4)/m_4, & \\
    &\mathcal{G}_{\theta_{63}} = (-\dot{x}_3+\dot{x}_4)/m_3,  &\mathcal{G}_{\theta_{64}} = (\dot{x}_3-\dot{x}_4)/m_4, & \\
\end{align} We show an example simulation of the displacement and velocities from the model in Figure \ref{fig:Example_FourSpring}. An averaged ranking of parameter influence, given by the sensitivity system across all 100 datasets, is provided in Figure \ref{fig:FourSpring_Sensitivity} of the Appendix (Sec. \ref{sec:AppSens}).

\begin{figure}
    \centering
    \includegraphics[width=0.45\linewidth]{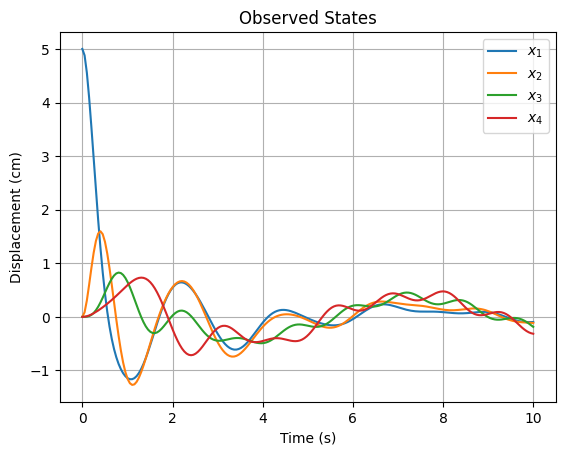}
    \includegraphics[width=0.45\linewidth]{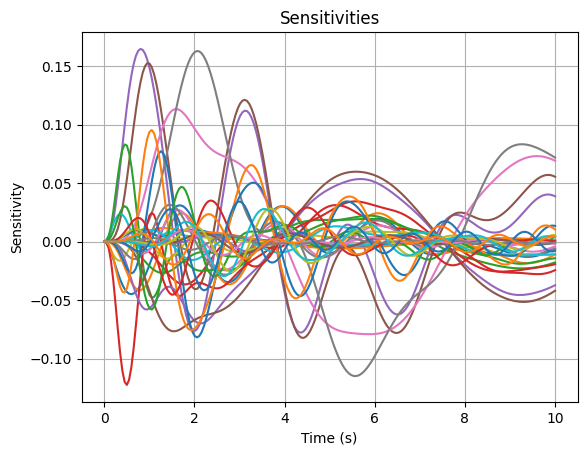}
    \caption{Example realization from the four-mass model. (Left): Example trajectory for the four mass displacements. (Right): The corresponding 32 sensitivities of each model state to the parameters of the system}
    \label{fig:Example_FourSpring}
\end{figure}

\subsubsection{Inference Results}
We again impose the measurement operator $A$ that maps the full displacement space for the model, which includes $x_1,x_2,x_3,$ and $x_4$. We again use 20 equally spaced time points on the interval $\Omega=[0,10]$. Each dataset is generated by sampling parameters from a uniform distribution, again using $k_i\sim \mathcal{U}(1,100)$, $i=1,2,3,4$, and $c_i\sim \mathcal{U}(0.1,10)$, $i=1,2,3,4$ as in the case of the two-mass system. We fixed the external force parameters for all simulations and data at $F_0=25$ and $\omega_f=5$. The initial conditions for the model are $x_1(0)=5.0$ and $x_2=x_3=x_4=0.0$, as well as $\dot{x}_1=\dot{x}_2=\dot{x}_3=\dot{x}_4=0.0$. The masses are fixed at $m_1=5.0, m_2=4.0,m_3=7.0$, and $m_4=10$. We add independent zero-mean Gaussian noise to each dataset, with a noise variance $\sigma_\varepsilon^2=0.1$. We use the same priors for the model discrepancy, with $\gamma \sim \textrm{LogNormal}\left(-2.0,1.0^2\right)$ and $\ell \sim \textrm{LogNormal}\left(-0.2,0.8^2\right)$, as was the case in the two-mass example.

Parameter inference results comparing KOH and the MOS approach can be found in Table \ref{tab:four_springs_summary_metrics}. As with the two-mass example, we find that our MOS approach achieves a smaller mean RMSE for nearly all parameters except for $c_4$ where KOH out performs, while the bias results are less consistently in favor of either approach. Both methods demonstrate empirical coverage that is quite similar across all of the parameters, including some severe under-coverage for several parameters, including $c_4,c_3$ and $k_3$. We attribute this to the relatively small influence on the model behavior, as the model outputs are less sensitive to these parameters (see Figure \ref{fig:FourSpring_Sensitivity} of the Appendix (Sec. \ref{sec:AppSens})).
\par 
We provide an example of the output uncertainty and posterior predictive in Figure \ref{fig:FourSpring_Output_UQ}. Error metrics in output space are provided in Table \ref{tab:FourSpring_Output_comparison}. The proposed MOS method results in lower MSE, empirical coverage closer to the nominal level, and smaller average intervals, though improvements are somewhat modest.

\begin{table}[ht]
\centering
\footnotesize
\caption{Posterior summary metrics for the four-mass oscillator across 100 Monte-Carlo experimental replications. Bold values for the bias and RMSE are the best for the given parameter.}
\label{tab:four_springs_summary_metrics}
\begin{tabular}{l|ccccc}
\hline
Method & Parameter & Mean Bias & Mean RMSE & Mean Coverage & Mean Int. Width \\
\hline
KOH & $k_1$ & \textbf{0.5984} & 5.3948 & 0.96 & 36.6688 \\
KOH & $c_1$ & 0.3017 & 1.1947 & 0.92 & 6.7745 \\
KOH & $k_2$ & -1.2837 & 8.8567 & 0.96 & 54.2071 \\
KOH & $c_2$ & \textbf{0.0739} & 1.3538 & 0.91 & 6.9601 \\
KOH & $k_3$ & \textbf{-2.8743} & 12.4060 & 0.88 & 63.0391 \\
KOH & $c_3$ & \textbf{0.6001} & 2.1844 & 0.86 & 8.5545 \\
KOH & $k_4$ & 2.8305 & 13.0305 & 0.96 & 68.8297 \\
KOH & $c_4$ & 0.5589 & \textbf{2.1696} & 0.78 & 8.1605 \\
\hline
MOS & $k_1$ & 0.6039 & \textbf{5.2327} & 0.96 & 34.0988 \\
MOS & $c_1$ & \textbf{0.2574} & \textbf{1.0994} & 0.94 & 6.4482 \\
MOS & $k_2$ & \textbf{-1.1730} & \textbf{8.0224} & 0.96 & 51.8563 \\
MOS & $c_2$ & 0.1166 & \textbf{1.2479} & 0.91 & 6.6844 \\
MOS & $k_3$ & -3.0939 & \textbf{11.8767} & 0.88 & 61.6692 \\
MOS & $c_3$ & 0.6302 & \textbf{2.0900 }& 0.85 & 8.4068 \\
MOS & $k_4$ & \textbf{2.8293} & \textbf{12.1778} & 0.97 & 66.0800 \\
MOS & $c_4$ &\textbf{ 0.5555} & 2.1859 & 0.80 & 8.1247 \\
\hline
\end{tabular}
\end{table}

\begin{table}[htbp]
  \centering
  \caption{Output uncertainty metrics between KOH and our MOS approach for the four-mass oscillator. Error metrics are first time-averaged and then averaged across 100 Monte-Carlo experimental replications. Bold values for the bias and RMSE are the best for the given parameter.}
  \label{tab:FourSpring_Output_comparison}
  \begin{tabular}{l|cc}
    \hline\hline
    & \textbf{KOH} & \textbf{MOS} \\
    \cline{2-3}
    \textbf{Metric} & Mean & Mean \\
    \hline
    Time Integrated RMSE           & 0.3025 & \textbf{0.2985} \\
    Time Integrated Bias           & -0.0060 & \textbf{-0.0060} \\
    Time Integrated Coverage       & 0.9715 & 0.9680 \\
    Time Integrated Interval Width & 1.7748 & 1.6764 \\
    \hline\hline
  \end{tabular}
\end{table}

% \begin{table}[htbp]
%   \centering
%   \caption{Output uncertainty metrics between KOH and Orthogonal approach for the four-mass oscillator. Error metrics are first time-averaged and then averaged across 100 Monte-Carlo experimental replications.}
%   \label{tab:FourSpring_Output_comparison}
%   \begin{tabular}{l|cc|cc}
%     \hline\hline
%     & \multicolumn{2}{c}{\textbf{KOH}} & \multicolumn{2}{c}{\textbf{MOS}} \\
%     \cline{2-3} \cline{4-5}
%     \textbf{Metric} & Mean & Median & Mean & Median \\
%     \hline
%     Time Integrated RMSE           & 0.3025 & 0.2651 & \textbf{0.2985} & \textbf{0.2641} \\
%     Time Integrated Bias           & -0.0060 & \textbf{-0.0026} & \textbf{-0.0060} & -0.0032 \\
%     Time Integrated Coverage       & 0.9715 & 1.0000 & 0.9680 & 1.0000 \\
%     Time Integrated Interval Width & 1.7748 & 1.6814 & 1.6764 & 1.5689 \\
%     \hline\hline
% \end{tabular}
% \end{table}

\begin{figure}
    \centering
    \includegraphics[width=\linewidth]{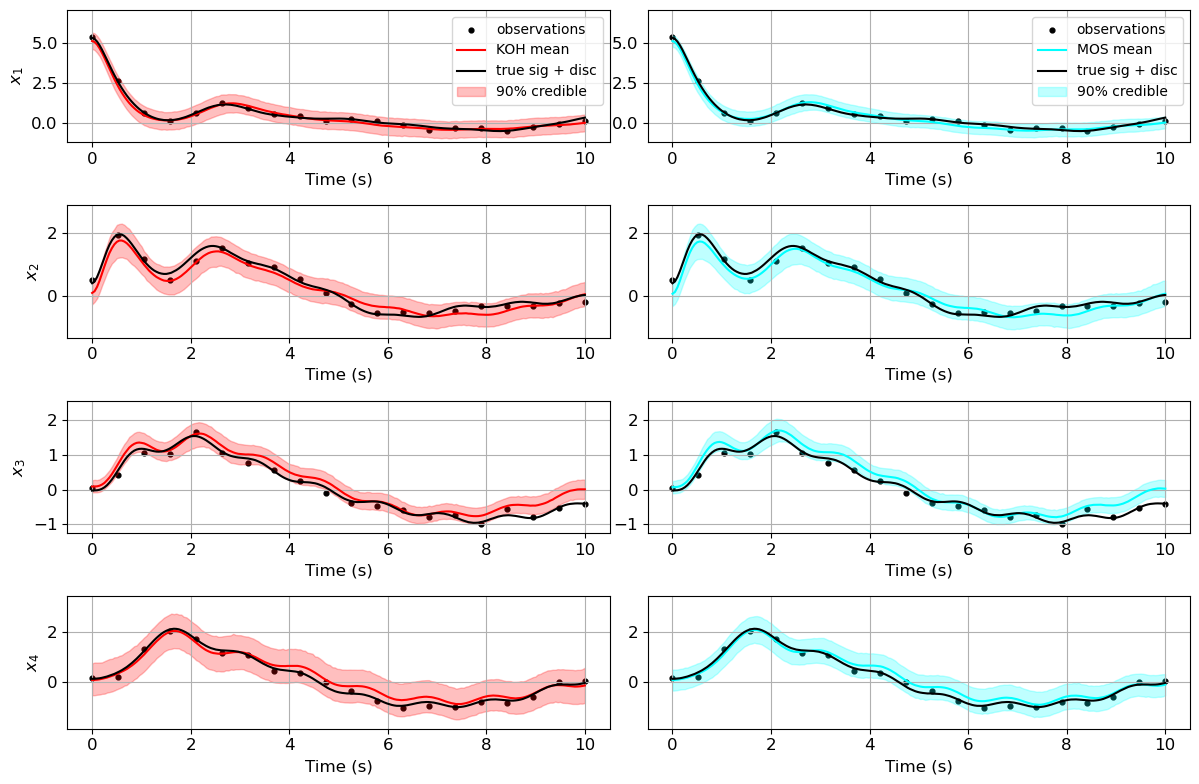}   

    \caption{Output UQ for four-mass oscillator using one representative dataset. Posterior samples of the simulator and discrepancy are drawn from either the KOH (left column, red) or MOS (right column, cyan) approach.}
    \label{fig:FourSpring_Output_UQ}
\end{figure}

\pagebreak

\subsection{Model 3: The FitzHugh-Nagumo model}

In this section, we consider the FitzHugh–Nagumo model, which is an oscillator-like equation that can describe the dynamics of excitatory and inhibitory neurons. The system of equations is given by
\begin{align}
    \frac{dv}{dt} &= v - \frac{v^3}{3} - w + I \\
    \frac{dw}{dt} &= (v+a-bw)\frac{1}{\tau}
\end{align}
where $v$ and $w$ represent the excitatory and inhibitory signals, respectively. The parameters are $a$ and $b$, which control the threshold for excitatory spiking and the resiting potential, an external applied current $I$, and the time-scale modulation $\tau$, which controls how the inhibitory signal acts relative to the excitatory signal. Thus, our parameters for inference are $\theta=\left\{a,b,I,\tau\right\}$. The sensitivity system is then controlled by $\mathcal{G}_z = \textrm{diag}\left(\mathcal{J}\right)$, where
\begin{equation}
    \mathcal{J} = \begin{bmatrix}
        1 - v^2 & -1 \\
        \displaystyle\frac{1}{\tau} & \displaystyle\frac{-b}{\tau}
    \end{bmatrix}
\end{equation}
and the vector $\mathcal{G}_\theta\in\mathbb{R}^8$ which has four non-trivial entries
\begin{align}
    &\mathcal{G}_{\theta_{3}} = 1,   &\mathcal{G}_{\theta_{5}} &= \frac{1}{\tau}, & \\
    &\mathcal{G}_{\theta_{6}} =  \frac{-1}{\tau},  &\mathcal{G}_{\theta_{8}} &= \frac{-(v+a-bw)}{\tau^2}.
\end{align}
Figure \ref{fig:Example_Fitz} shows a sample trajectory and the corresponding sensitivities for the system. An averaged ranking of parameter influence, given by the sensitivity system across all 100 datasets, is provided in Figure \ref{fig:Fitz_Sensitivity} of the Appendix (Sec. \ref{sec:AppSens}).

\begin{figure}[ht!]
    \centering
    \includegraphics[width=0.45\linewidth]{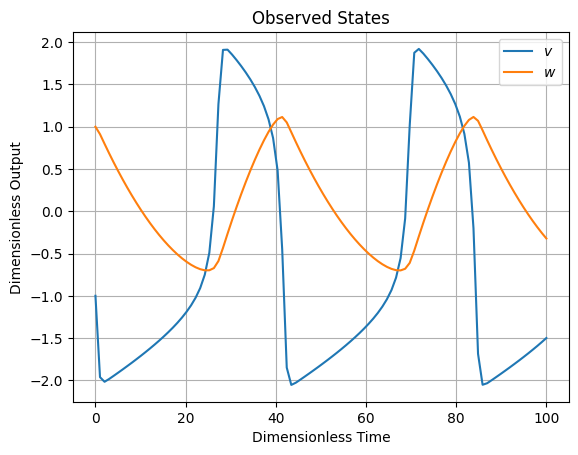}
    \includegraphics[width=0.45\linewidth]{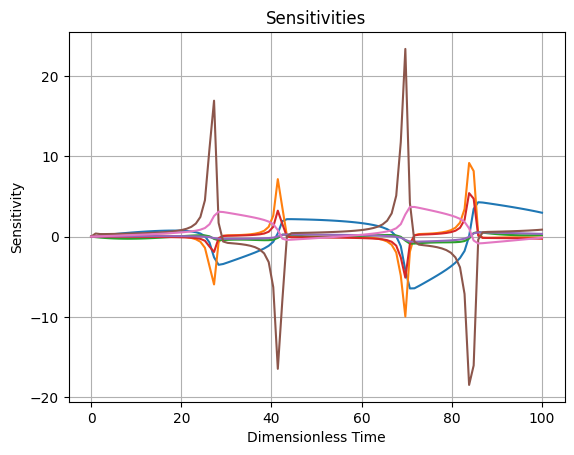}
    \caption{Example realization from the FitzHugh-Nagumo model. (Left) Example trajectory of the nondimensional $v$ and $w$ state variables. (Right) Corresponding $8$ dimensional sensitivities.}
    \label{fig:Example_Fitz}
\end{figure}

\subsubsection{Inference Results}
We let our measurement operator $A$ map the full state space for the model, which includes the nondimensional states $v$ and $w$. We use 20 equally spaced time points on the interval $\Omega=[0,100]$. Each dataset is generated by sampling parameters from uniform distributions, with $a\sim \mathcal{U}(0.1,1)$, $b\sim \mathcal{U}(0.1,10)$, $I\sim \mathcal{U}(0.1,1)$, and $\tau \sim \mathcal{U}(0.04,0.1)$. The initial conditions for the model are $v(0)=-1.0$ and $w(0)=1.0$. We add independent zero-mean Gaussian noise to each dataset, with a noise variance $\sigma_\varepsilon^2=0.1$. For the model discrepancy, we generate a signal using the MOS framework with hyperparameters $\gamma \sim \textrm{LogNormal}\left(-1.6,1.0^2\right)$ and $\ell \sim \textrm{LogNormal}\left(3,0.4^2\right)$. We use these priors for the GP to ensure a relatively small discrepancy magnitude (since the states are non-dimensional) and a somewhat large length-scale to match the time domain $\Omega$.

Results comparing KOH and the orthogonal MOS approach can be found in Table \ref{tab:fitz_summary_metrics}. For parameter estimation, results echo previous models with uniformly lower MSE for MOS compared to KOH. The empirical coverage is similar for both methods, and severely undercover for the parameter $b$, while attaining reasonable coverage for the remaining parameter. The mean interval width for MOS is again uniformly lower for all parameters. We note that the average sensitivity of the model to $b$ and $I$ is much lower than $\tau$ and $a$ parameters (see Figure~\ref{fig:Fitz_Sensitivity} of the Appendix (Sec. \ref{sec:AppSens}).
 
Figure \ref{fig:Fitz_Output_UQ} shows example of the predictive posteriors in the output-space, with error metrics and UQ in Table \ref{tab:Fitz_Output_comparison}. MOS demonstrates lower MSE and tighter posterior intervals than KOH, however, both method produce intervals that are overly conservative, with empirical coverage of nominal $95\%$ intervals near $100\%$.

\begin{table}[ht]
\centering
\footnotesize
\caption{Posterior summary metrics for the FitzHugh-Nagumo model across 100 Monte-Carlo experimental replications. Bold values for the bias and RMSE are the best for the given parameter.}
\label{tab:fitz_summary_metrics}
\begin{tabular}{llrrrr}
\hline
Method & Parameter & Mean Bias & Mean RMSE & Mean Coverage & Mean Int. Width \\
\hline
KOH & $a$ & \textbf{-0.0025} & 0.0833 & 0.92 & 0.4795 \\
KOH & $b$ & 0.0059 & 0.1541 & 0.83 & 0.6305 \\
KOH & $I$ & 0.2772 & 1.1208 & 0.90 & 4.9521 \\
KOH & $\tau$ & \textbf{0.0137} & 0.0838 & 0.95 & 0.4623 \\
\hline
MOS & $a$ & -0.0037 & \textbf{0.0655} & 0.91 & 0.3876 \\
MOS & $b$ & \textbf{0.0021} & \textbf{0.1443 }& 0.84 & 0.5883 \\
MOS & $I$ & \textbf{0.2370} & \textbf{1.0393} & 0.90 & 4.3114 \\
MOS & $\tau$ & 0.0143 & \textbf{0.0424} & 0.97 & 0.3238 \\
\hline
\end{tabular}
\end{table}

\begin{table}[htbp]
  \centering
  \caption{Output uncertainty metrics between KOH and our MOS approach for the FitzHugh-Nagumo model. Error metrics are first time-averaged and then averaged across 100 Monte-Carlo experimental replications. Bold values for the bias and RMSE are the best for the given parameter.}
  \label{tab:Fitz_Output_comparison}
  \begin{tabular}{l|cc}
    \hline\hline
    & \textbf{KOH} & \textbf{MOS} \\
    \cline{2-3}
    \textbf{Metric} & Mean & Mean \\
    \hline
    Time Integrated RMSE           & 0.2952 & \textbf{0.2887} \\
    Time Integrated Bias           & 0.0056 & \textbf{0.0052} \\
    Time Integrated Coverage       & 0.9999 & 1.0000 \\
    Time Integrated Interval Width & 3.2241 & 2.6823 \\
    \hline\hline
  \end{tabular}
\end{table}

% \begin{table}[htbp]
%   \centering
%   \caption{Output uncertainty metrics between KOH and Orthogonal approach for the FitzHugh-Nagumo model. Error metrics are first time-averaged and then averaged across 100 Monte-Carlo experimental replications.}
%   \label{tab:Fitz_Output_comparison}
%   \begin{tabular}{l|cc|cc}
%     \hline\hline
%     & \multicolumn{2}{c}{\textbf{KOH}} & \multicolumn{2}{c}{\textbf{MOS}} \\
%     \cline{2-3} \cline{4-5}
%     \textbf{Metric} & Mean & Median & Mean & Median \\
%     \hline
%     Time Integrated RMSE           &  0.2952 &  0.2304 &  \textbf{0.2887} &  \textbf{0.2239} \\
%     Time Integrated Bias           &  0.0056 & -0.0034 &  \textbf{0.0052} &  \textbf{-0.0004} \\
%     Time Integrated Coverage       &  0.9999 &  1.0000 &  1.0000 &  1.0000 \\
%     Time Integrated Interval Width &  3.2241 &  2.7405 &  2.6823 &  2.2609 \\
%     \hline\hline
% \end{tabular}
% \end{table}

\begin{figure}
    \centering
    \includegraphics[width=\linewidth]{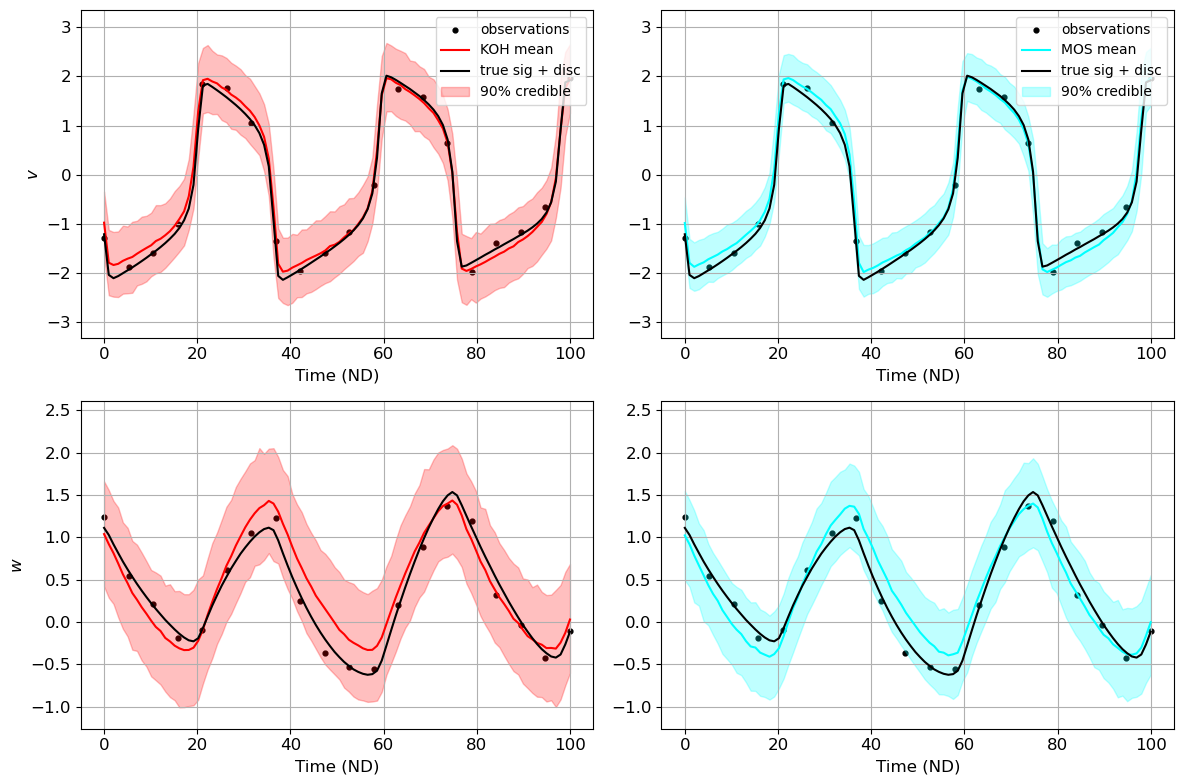}        
    \caption{Output UQ for FitzHugh-Nagumo model using one representative dataset. Posterior samples of the simulator and discrepancy are drawn from either the KOH (left column, red) or MOS (right column, cyan) approach.}
    \label{fig:Fitz_Output_UQ}
\end{figure}

\newpage

\subsection{Model 4: The susceptible-infected-recovered-vaccinated (SIRV) disease model}
The susceptible-infected-recovered-vaccinated (SIRV) model is a system of nonlinear ODEs given by
\begin{equation}\label{eqn:SIRV}
\begin{aligned}
    \frac{dS}{dt} &= -\alpha SI/N-vS+{w_r}R+w_{v}V\\
    \frac{dI}{dt} &=  \alpha SI/N - \mu I  \\
    \frac{dR}{dt} &= \mu I - w_r R\\
    \frac{dV}{dt} &= vS - w_v V
\end{aligned}
\end{equation}
where $\alpha$ is the infection rate, $v$ is the vaccination rate, $\mu$ represents the transition from infected to recovered, $w_{r}$ is the rate of diminished immunity that leads to being susceptible again, and $w_{v}$ reflects the diminished effects of vaccination. The term $N$ is the total population size, where we assume $N=S+I+V+R$ is a constant. The state vector $(S(t),I(t),R(t),V(t))^{\intercal}$ denotes the susceptible, infected and recovered populations at time $t$, respectively.  Thus, the full parameter set is $\theta=\left\{\alpha,v,\mu,w_r,w_v\right\}$. These models have been used extensively to study the spread of pathogens such as measles, dengue virus, and, more recently, COVID-19 \citep{ElHajji2022,Hollingsworth2020,Wang2026}.

The sensitivity system of the SIRV model can be derived with knowledge of $\mathcal{J}$ and $\mathcal{G}_\theta$. The Jacobian, $\mathcal{J}$ is
\begin{align}
    \mathcal{J}   =
    \begin{bmatrix} 
    -\alpha I/N - v & -\alpha S/{N} & w_r & w_v \\[6pt]
\alpha I/{N} & \alpha S/{N} - \mu & 0 & 0 \\[6pt]
0 & \mu & -w_r & 0 \\[6pt]
v & 0 & 0 & -w_v
    \end{bmatrix}
\end{align}
and $\mathcal{G}_\theta\in\mathbb{R}^{20}$ has nonzero entries
\begin{align}
    \mathcal{G}_{\theta_1} = -SI/{N}, \ \  \mathcal{G}_{\theta_2} = -S, \ \  
    \mathcal{G}_{\theta_4} = R, \ \ 
    \mathcal{G}_{\theta_5} = V, \ \
    \mathcal{G}_{\theta_6} = SI/N, \nonumber \\ \nonumber \\
    \mathcal{G}_{\theta_8} = -I, \ \ 
    \mathcal{G}_{\theta_{13}} = I, \ \ 
    \mathcal{G}_{\theta_{14}} = -R, \ \
    \mathcal{G}_{\theta_{17}} = S, \ \ 
    \mathcal{G}_{\theta_{20}} = -V. 
\end{align}

For this model, we consider two different observation operators. The first observation operator, $A_{IRV}$ returns the $I,R,V$ states, while the second observation operator, $A_{IV}$, returns only the $I$ and $V$ states of the model. These are defined as the matrices
$$
A_{IRV} = \begin{pmatrix}
    0 & 0 & 0 & 0 \\
    0 & 1 & 0 & 0 \\
    0 & 0 & 1 & 0 \\
    0 & 0 & 0 & 1 \\
\end{pmatrix}
\ \ \ \ \ 
A_{IV} = \begin{pmatrix}
    0 & 0 & 0 & 0 \\
    0 & 1 & 0 & 0 \\
    0 & 0 & 0 & 0 \\
    0 & 0 & 0 & 1 \\
\end{pmatrix}.
$$
We note that this observation operator is consistent with previous studies using such an epidemic model \cite{Tavener2011,Wang2026}, as susceptible individuals ($S$) are typically difficult to obtain for model calibration. We also scaled predictions by the total population $N$ so that observations represent percentages. An averaged ranking of parameter influence, given by the sensitivity system across all 100 datasets, is provided in Figure \ref{fig:SIRV_Sensitivity} of the Appendix (Sec. \ref{sec:AppSens}).

% \begin{equation*}
%     \mathcal{G}_\theta =  \begin{bmatrix}
%        -\dfrac{z_1 z_2}{N} \\ -z_1 \\ 0 \\ z_3 \\ z_4 \\ \dfrac{z_1 z_2}{N} \\ 0 \\ -z_2 \\ 0 \\ 0 \\ 0 \\ 0 \\ z_2 \\ -z_3 \\ 0 \\ 0 \\ z_1 \\ 0 \\ 0 \\ -z_4 \end{bmatrix},
% \end{equation*}

\begin{figure}
    \centering
    \includegraphics[width=0.45\linewidth]{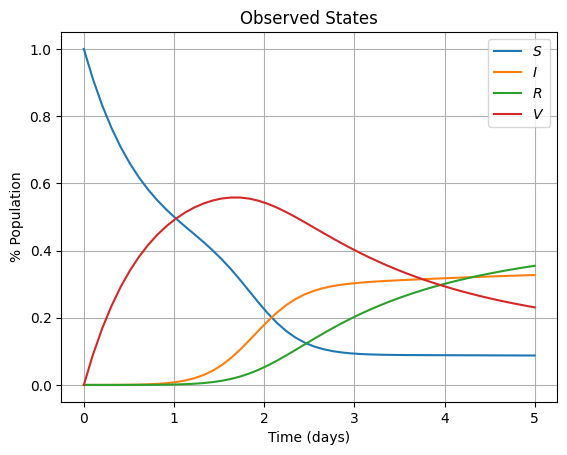}
    \includegraphics[width=0.45\linewidth]{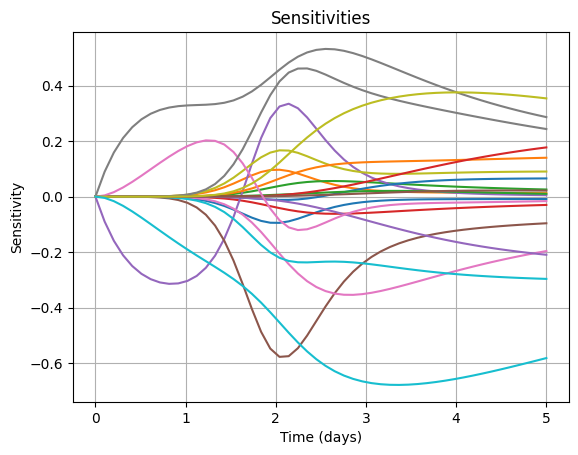}
    \caption{Example realization from the $SIRV$ model. (Left): Example trajectory for the $S,I,R,$ and $V$ states. (Right): The corresponding 30 sensitivities of each model state to the parameters of the system}
    \label{fig:Example_SIRV}
\end{figure}

\subsubsection{Inference Results}
We use 20 equally spaced time points on the interval $\Omega=[0,5]$, representing five days. Each dataset is generated by sampling parameters from uniform distributions, with $a\sim \mathcal{U}(8,16)$, $v\sim \mathcal{U}(0.1,1)$, $\mu\sim \mathcal{U}(0.1,1)$, $w_r\sim \mathcal{U}(0.1,1)$, and $w_v \sim \mathcal{U}(0.1,1)$. The initial conditions for the model are $S(0)=0.99999$, $I=0.00001$, and $V(0)=R(0)=0$. We add independent zero-mean Gaussian noise to each dataset, with a noise variance $\sigma_\varepsilon^2=0.001$. For the model discrepancy, we generate a signal using the MOS framework with hyperparameters $\gamma \sim \textrm{LogNormal}\left(-7,1.0^2\right)$ and $\ell \sim \textrm{LogNormal}\left(-0.5,0.8^2\right)$. We use these priors for the GP to ensure a relatively small discrepancy magnitude (since the states are scaled by the population size) and a somewhat large length-scale that allows for oscillations slowly across the first few days of the simulation.

Results comparing KOH and the orthogonal MOS approach for the operator $A_{IRV}$ can be found in Table \ref{tab:sirv_irv_summary_metrics}. Parameter estimation results are similar to previous models with consistently lower bias and RMSE for MOS compared to KOH. The MOS approach also leads to slightly lower mean interval widths and improved mean coverage for a non-influential parameters $w_r$ (see sensitivities in Figure \ref{fig:SIRV_Sensitivity} of the Appendix (Sec. \ref{sec:AppSens})).

Figure \ref{fig:IRV_Output_UQ} shows example of the predictive posteriors in the output-space, with error metrics and uncertainty quantification in Table \ref{tab:irv_output_metrics}. MOS demonstrates lower RMSE and tighter posterior intervals than KOH, but the time integrated bias is slightly better for KOH. The KOH method produce intervals that are overly conservative, with empirical coverage of nominal $95\%$ intervals near $100\%$. The average emperical coverage for the MOS approach is approximately 95\%.

\begin{table}[ht]
\centering
\footnotesize
\caption{Posterior summary metrics for the $SIRV$ model across 100 Monte-Carlo experimental replications using the $A_{IRV}$ operator. Bold values for the bias and RMSE are the best for the given parameter.}
\label{tab:sirv_irv_summary_metrics}
\begin{tabular}{llrrrr}
\hline
Method & Parameter & Mean Bias & Mean RMSE & Mean Coverage & Mean Int. Width \\
\hline
KOH & $a$ & 0.0410 & 0.2044 & 0.99 & 2.1091 \\
KOH & $v$ & 0.0120 & 0.0310 & 0.97 & 0.2874 \\
KOH & $\mu$ & 0.0010 & 0.0486 & 0.96 & 0.3133 \\
KOH & $w_r$ & 0.0070 & 0.1124 & 0.88 & 0.5384 \\
KOH & $w_v$ & 0.0260 & 0.0551 & 0.97 & 0.3987 \\
\hline
MOS & $a$ & \textbf{0.0291} & \textbf{0.1324} & 0.99 & 1.5453 \\
MOS & $v$ & \textbf{0.0069} & \textbf{0.0187} & 0.97 & 0.2000 \\
MOS & $\mu$ & \textbf{-0.0007} & \textbf{0.0316} & 0.96 & 0.2169 \\
MOS & $w_r$ & \textbf{0.0061} & \textbf{0.0771} & 0.89 & 0.4221 \\
MOS & $w_v$ & \textbf{0.0185} & \textbf{0.0352} & 0.97 & 0.2760 \\
\hline
\end{tabular}
\end{table}

\begin{table}[htbp]
  \centering
  \caption{Output uncertainty metrics between KOH and our MOS approach for the $SIRV$ model using the $A_{IRV}$ operator. Error metrics are first time-averaged and then averaged across 100 Monte-Carlo experimental replications. Bold values for the bias and RMSE are the best for the given parameter.}
  \label{tab:irv_output_metrics}
  \begin{tabular}{l|cc}
    \hline\hline
    & \textbf{KOH} & \textbf{MOS} \\
    \cline{2-3}
    \textbf{Metric} & Mean & Mean \\
    \hline
    Time Integrated RMSE           & 0.0194 & \textbf{0.0187} \\
    Time Integrated Bias           & \textbf{-0.0004} & -0.0005 \\
    Time Integrated Coverage       & 0.9888 & 0.9547 \\
    Time Integrated Interval Width & 0.1657 & 0.1082 \\
    \hline\hline
  \end{tabular}
\end{table}

% \begin{table}[htbp]
%   \centering
%   \caption{Output uncertainty metrics between KOH and Orthogonal approach for the $SIRV$ model with observations of the $I,R,$ and $V$ states. Error metrics are first time-averaged and then averaged across 100 Monte-Carlo experimental replications.}
%   \label{tab:irv_output_metrics}
%   \begin{tabular}{l|cc|cc}
%     \hline\hline
%     & \multicolumn{2}{c}{\textbf{KOH}} & \multicolumn{2}{c}{\textbf{MOS}} \\
%     \cline{2-3} \cline{4-5}
%     \textbf{Metric} & Mean & Median & Mean & Median \\
%     \hline
%     Time Integrated RMSE           & 0.0194 & 0.0158 & \textbf{0.0187} & \textbf{0.0152} \\
%     Time Integrated Bias           & \textbf{-0.0004} & \textbf{-0.0000} & -0.0005 & 0.0002 \\
%     Time Integrated Coverage       & 0.9888 & 1.0000 & 0.9547 & 1.0000 \\
%     Time Integrated Interval Width & 0.1657 & 0.1397 & 0.1082 & 0.0909 \\
%     \hline\hline
% \end{tabular}
% \end{table}

\begin{figure}
    \centering
    \includegraphics[width=\linewidth]{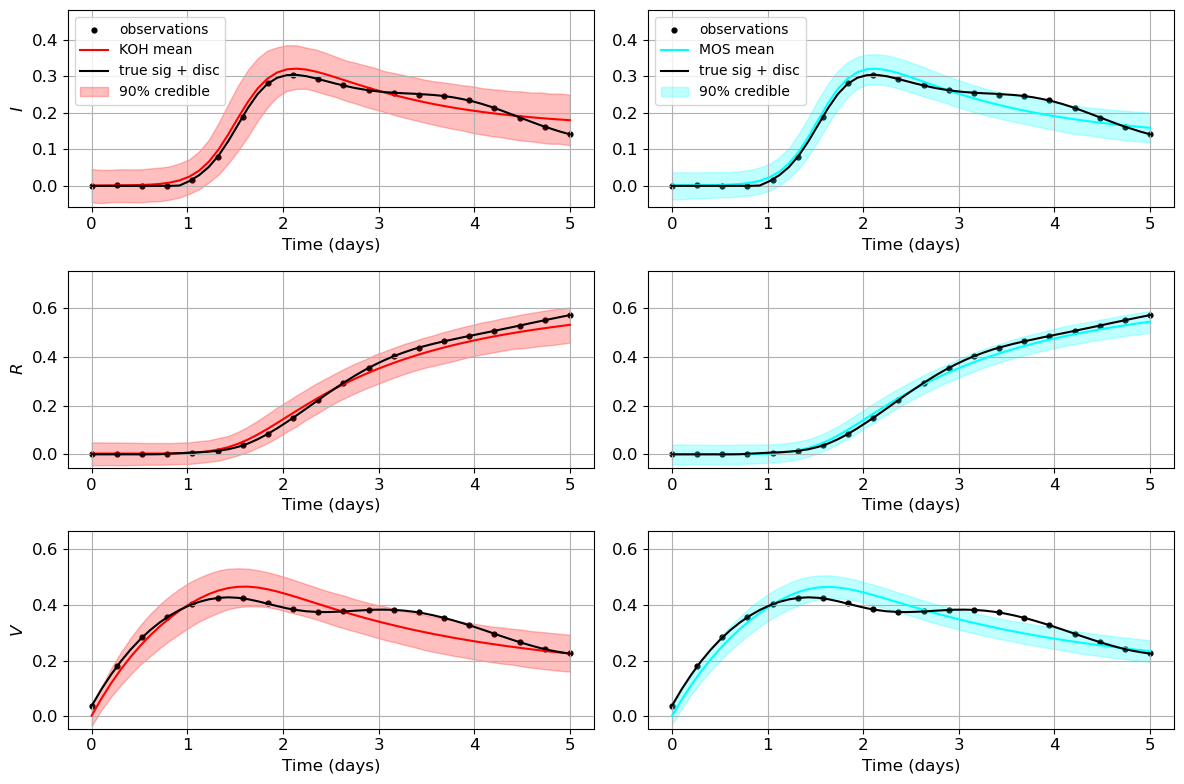}     
     \caption{Output UQ for $SIRV$ model with the $A_{IRV}$ observation operator using one representative dataset. Posterior samples of the simulator and discrepancy are drawn from either the KOH (left column, red) or MOS (right column, cyan) approach.}
    \label{fig:IRV_Output_UQ}
\end{figure}

We generate data in an identical fashion to the previous section, with the same parameter and GP discrepancy priors as before, but using the observation operator $A_{IV}$. Results for this design can be found in Table \ref{tab:sirv_iv_summary_metrics}. Parameter estimation results are similar to when the observation operator $A_{IRV}$ is used, but with KOH providing a slightly better mean bias in the parameter $a$ than the MOS method. We find that this observation operator, $A_{IV}$, leads to a lower RMSE in general for the estimation of $a$, but also leads to a reduction in coverage for the parameter $w_r$. We see in general that again the MOS approach provides nearly similar or slightly worse coverage but consistently leads to smaller interval widths and reduced uncertainty.

Figure \ref{fig:IV_Output_UQ} shows example of the predictive posteriors in the output-space, with error metrics and uncertainty quantification in Table \ref{tab:iv_output_metrics}. MOS demonstrates lower RMSE and tighter posterior intervals than KOH, but the time integrated bias is slightly better for KOH. The KOH method produce intervals that are overly conservative, with empirical coverage of nominal $95\%$ intervals near $100\%$. The average emperical coverage for the MOS approach is approximately 95\%. We find again that output uncertainty is reduced in the MOS approach, and that this method tends to find output-specific uncertainty, whereas KOH tends to apply the same amount of uncertainty across state variables.

\begin{table}[ht]
\centering
\footnotesize
\caption{Posterior summary metrics for the $SIRV$ model across 100 Monte-Carlo experimental replications using the $A_{IV}$ operator. Bold values for the bias and RMSE are the best for the given parameter.}
\label{tab:sirv_iv_summary_metrics}
\begin{tabular}{llrrrr}
\hline
Method & Parameter & Mean Bias & Mean RMSE & Mean Coverage & Mean Int. Width \\
\hline
KOH & $a$ & \textbf{0.0108} & 0.1770 & 0.97 & 1.9289 \\
KOH & $v$ & 0.0050 & 0.0251 & 0.98 & 0.2709 \\
KOH & $\mu$ & 0.0066 & 0.0568 & 0.97 & 0.3486 \\
KOH & $w_r$ & 0.0114 & 0.1470 & 0.78 & 0.6186 \\
KOH & $w_v$ & 0.0158 & 0.0530 & 0.93 & 0.4002 \\
\hline
MOS & $a$ & 0.0115 & \textbf{0.1239} & 0.98 & 1.2471 \\
MOS & $v$ & \textbf{0.0034} & \textbf{0.0148} & 0.98 & 0.1540 \\
MOS & $\mu$ & \textbf{0.0022} &\textbf{ 0.0337} & 0.96 & 0.2192 \\
MOS & $w_r$ & \textbf{0.0064} & \textbf{0.1050} & 0.80 & 0.4732 \\
MOS & $w_v$ & \textbf{0.0102} & \textbf{0.0294} & 0.95 & 0.2354 \\
\hline
\end{tabular}
\end{table}

\begin{table}[htbp]
  \centering
  \caption{Output uncertainty metrics between KOH and our MOS approach for the $SIRV$ model using the $A_{IV}$ operator. Error metrics are first time-averaged and then averaged across 100 Monte-Carlo experimental replications. Bold values for the bias and RMSE are the best for the given parameter.}
  \label{tab:iv_output_metrics}
  \begin{tabular}{l|cc}
    \hline\hline
    & \textbf{KOH} & \textbf{MOS} \\
    \cline{2-3}
    \textbf{Metric} & Mean & Mean \\
    \hline
    Time Integrated RMSE           & 0.0173 & \textbf{0.0159} \\
    Time Integrated Bias           & -0.0015 & \textbf{-0.0014} \\
    Time Integrated Coverage       & 1.0000 & 1.0000 \\
    Time Integrated Interval Width & 0.2533 & 0.1885 \\
    \hline\hline
  \end{tabular}
\end{table}

% \begin{table}[htbp]
%   \centering
%   \caption{Output uncertainty metrics between KOH and our MOS approach for the $SIRV$ model with observations of the $I$ and $V$ states. Error metrics are first time-averaged and then averaged across 100 Monte-Carlo experimental replications.}
%   \label{tab:iv_output_metrics}
%  \begin{tabular}{l|cc|cc}
%     \hline\hline
%     & \multicolumn{2}{c}{\textbf{KOH}} & \multicolumn{2}{c}{\textbf{MOS}} \\
%     \cline{2-3} \cline{4-5}
%     \textbf{Metric} & Mean & Median & Mean & Median \\
%     \hline
%     Time Integrated RMSE           & 0.0173 & 0.0140 & \textbf{0.0159} & \textbf{0.0126} \\
%     Time Integrated Bias           & -0.0015 & -0.0009 & \textbf{-0.0014} & \textbf{-0.0007} \\
%     Time Integrated Coverage       & 1.0000 & 1.0000 & 1.0000 & 1.0000 \\
%     Time Integrated Interval Width & 0.2533 & 0.2278 & 0.1885 & 0.1654 \\
%     \hline\hline
% \end{tabular}
% \end{table}

\begin{figure}
    \centering
    \includegraphics[width=\linewidth]{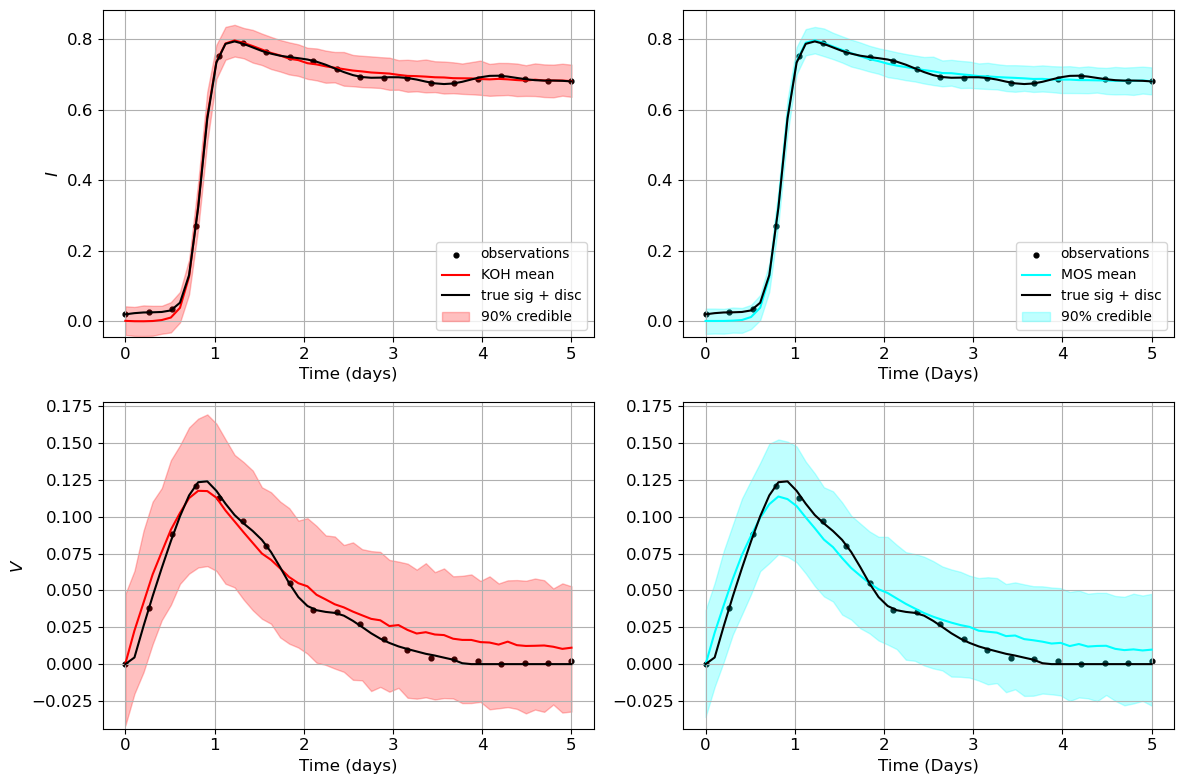}       
    \caption{Output UQ for $SIRV$ model with the $A_{IV}$ observation operator using one representative dataset. Posterior samples of the simulator and discrepancy are drawn from either the KOH (left column, red) or MOS (right column, cyan) approach.}
    \label{fig:IV_Output_UQ}
\end{figure}

\pagebreak

\pagebreak

\section{Conclusion and Future Work}\label{sec:Discussion}
Model discrepancy is a significant source of uncertainty in the case of simulation models \citep{hart2026toward,Kimpton2025,Paun2020}, yet this form of uncertainty is sparsely considered in physics and biological models governed by complex ODE models. To address this gap, we provide a novel approach to multi-output ODE model calibration under model discrepancy that is demonstrated to improve upon the dominant existing approach in several key areas.  %including estimation and uncertainty quantification of both model parameters and discrepancy-augmented model predictions. 
Central to our approach is a new discrepancy prior, semi-automatically elicited from the sensitivity structure of the forward model, that extends the orthogonal-prior construction of \cite{Plumlee03072017} to multivariate ODE calibration. %Differential equation models are ideal for such a framework, given that numerous applications uses nonlinear differential equation systems to describe complex, observable dynamics \citep{Colebank2024G,Zi2011,Kimpton2025}.
Extensive experimental evaluation demonstrates that our proposed MOS approach outperforms the standard method of multivariate calibration under model discrepancy (KOH). In particular, the MOS approach consistently results in lower estimation errors and largely provides well-calibrated intervals for (identifiable) model parameters and process predictions. This is anticipated, as the use of sensitivity-based priors on the discrepancy helps enforce a key structural identifiability condition on the target parameter. %Importantly, we find that this reduction in interval width does not lead to a drastic reduction in coverage in the parameters. In addition, the RMSE in model parameters is generally smaller on average for all model types, while the average bias is less consistent. This suggests that the MOS approach provides reduced uncertainty and, on average, more accurate parameter estimates compared to the $L^2$ minimizer.

%We can also draw similar conclusions in output space across all model types. We again find that the MOS approach provides a lower average time-integrated RMSE across all model types, while the average time-integrated bias is comparable but sometimes slightly larger than the KOH. Across all models, MOS provides a lower average time integrated interval width than KOH while till obtained a 95\% coverage on average. The sample predictions provided in our results echo this qualitatively. We tend to see that KOH has larger predictive uncertainty than MOS, even when it appears that the model can replicate the dynamics of the system well.

Our approach is not without several limitations. First, this specific approach leverages analytical sensitivities for ODE models. This would be limited for systems with parameters and/or state variables that in very high dimensions. Deriving sensitivity equations can be timely and complex in this setting. In addition, the computational complexity of solving the sensitivity system simultaneously with the ODE can be burdensome with high-dimensional parameter vectors. Potential solutions to this, such as the use of automatic differentiation and Jacobian-vector products in packages like Jax \citep{jax2018github}, should be investigated next. We generated data for our study from the models and MOS approach, hence our results are not for actual measurements that might suffer from even larger model discrepancy. Additional investigations into the selection of Gaussian process kernels is also warranted, as specific kernels may be better suited for certain simulations based on their sensitivity systems. %Finally, our approach is well suited for ODE systems, yet its extension to infinite-dimensional parameters and random fields, common in partial differential equation (PDE) models, is unknown. Moving forward, more effort should be focused on the use of this methodology in more expensive models with more complex parameter dependencies, which invites the use of surrogate modeling.  
% MJC - Contextualize results here and do less comparison perse

\section*{Acknowledgements}
This work was funded in part by a South Carolina EPSCoR Grant for Applications in Industry and Networking Program (26-GA01), through the University of South Carolina’s NIAID-funded R25AI164581 (R25) BigDataHealth Science Fellow Program, and through an American Heart Association Career Development Award (26CDA1599801, \\ doi:10.58275/AHA.26CDA1599801.pc.gr.243168).

\bibliographystyle{siam}
\bibliography{refs}

\appendix
\section{Proofs}
\noindent{\textbf{Lemma~\ref{lem:1st_order_orthogonality}}}:
\begin{proof}
Let $a_j$ be a column vector of $k$'th row of $A$. Consider that for each component of the loss and for any $\eta \in \Theta$:
$$
\begin{aligned}
    \|\xi_j(\cdot) - a_{j}^{\intercal}z(\cdot,\eta)\|_{L^2(\Omega)}^2 &=  \|a_{j}^{\intercal}z(\cdot, \theta) + \delta_j(\cdot, \theta) - a_{j}^{\intercal}z(\cdot,\eta)\|_{L^2(\Omega)}^2 \\
     &= \|\delta_j(\cdot,\theta)\|_{L^2(\Omega)}^2 + 2\langle \delta_j(\cdot|\theta), a_{j}^{\intercal}(z(\cdot, \theta) - z(\cdot,\eta)\rangle_{L^2(\Omega)} \\
     &\quad\quad + \|a_{j}^{\intercal}(z(\cdot, \theta) - z(\cdot,\eta))\|_{L^2(\Omega)}^2.
\end{aligned}
$$
The above decomposition implies we have the form of the gradient of the $L^2$-calibration loss for the $j$'th output dimension:
$$
\begin{aligned}
    \nabla_{\eta}\|\xi_j(\cdot) - a_{j}^{\intercal}z(\cdot,\eta)\|_{L^2(\Omega)}^2 &= \nabla_{\eta}\|\delta_j(\cdot;\theta)\|_{L^2(\Omega)}^2 + 2\nabla_{\eta}\langle \delta_j(\cdot;\theta), a_{j}^{\intercal}(z(\cdot, \theta) - z(\cdot,\eta))\rangle_{L^2(\Omega)}\\
     &\quad\quad + \nabla_{\eta}\|a_{j}^{\intercal}(z(\cdot, \theta) - z(\cdot,\eta))\|_{L^2(\Omega)}^2 \\
     &= -2\langle \delta_j(\cdot;\theta), a_{j}^{\intercal}\mathcal{S}(\cdot, \eta)\rangle_{L^2(\Omega)} \\
     &-2\langle a_{j}^{\intercal}z(\cdot, \theta), a_{j}^{\intercal}\mathcal{S}(\cdot, \eta)\rangle_{L^2(\Omega)} + 2\langle a_{j}^{\intercal}z(\cdot, \eta), a_{j}^{\intercal}\mathcal{S}(\cdot,\eta)\rangle_{L^2(\Omega)} 
\end{aligned}
$$
Evaluating the above gradient at $\theta$, the bottom two terms in the final equality cancel. Plugging in this expression for the $j$'th component into the full direct sum we have
\begin{equation}\label{eqn:multivariate_orthogonality_condition}
\begin{aligned}
        \nabla_{\eta=\theta} \|\xi(\cdot) - Az(\cdot, \theta) \|_{\mathcal{H}}^2 &= \nabla_{\eta=\theta}\sum_{k=1}^J\|\xi_k(\cdot) - a_k^{\intercal}z(\cdot,\theta)\|_{L^2(\Omega)}^2 \\
        &= 
        -2\sum_{j=1}^J\langle \delta_k(\cdot;\theta), a_k^{\intercal}\mathcal{S}(\cdot,\theta)\rangle_{L^2(\Omega)}^2\\
\end{aligned}
\end{equation}
Provided $\theta$ is in the interior of $\Theta$, the first order optimality conditions of $\theta$ coupled with \eqref{eqn:multivariate_orthogonality_condition} imply 
$$
\begin{aligned}
    0 & =\int_{\Omega} \mathcal{S}^{\intercal}(t, \theta)A^{\intercal}\delta(t;\theta)
\end{aligned}
$$
where $0\in\mathbb{R}^P$ denotes the vector of zeros, as desired to show. 
\end{proof}

\noindent{\textbf{Theorem~\ref{thrm:constrained_multivariate_prior}}:} 
\begin{proof}
Denote integrand of the left-hand side of \eqref{eqn:jac_subspace_oth} as $f_{\theta}(t) := S_{z}(t,\theta)^{\intercal}A^{\intercal}\delta_{\theta}$. $f_{\theta}(t)$ is a Gaussian process due to closure under linear operators. Hence, to show \eqref{eqn:jac_subspace_oth} holds with probability 1, its suffices to show that 
$\mathbb{E}[\int_{\Omega}f_{\theta}]=0\in\mathbb{R}^{p}$ and $\text{Diag}\{\text{Var}[\int_{\Omega} f_{\theta}]\}=0\in\mathbb{R}^{p}$
The zero mean condition of $\delta_{\theta}$ implies $f_{\theta}$ is mean zero and so the zero mean condition follows trivially. We now consider the variance. 
$$
\begin{aligned}
    \text{Var}[\int_{\Omega} f_{\theta}] &= \mathbb{E}[\left(\int_{\Omega}S_{z}(t,\theta)^{\intercal}A^{\intercal}\delta_{\theta}(t)dt\right)\left(\int_{\Omega}S_{z}(t^{\prime},\theta)^{\intercal}A^{\intercal}\delta_{\theta}(t^{\prime})dt^{\prime}\right)^{\intercal}] \\ 
    &= \int_{\Omega}\int_{\Omega}S_{z}(t,\theta)^{\intercal}A^{\intercal}\mathbb{E}[\delta_{\theta}(t)\delta_{\theta}^{\intercal}(t^{\prime})]A S_{z}(t^{\prime},\theta)dtdt^{\prime}\\
    &= \int_{\Omega}\int_{\Omega}S_{z}(t,\theta)^{\intercal}A^{\intercal}C_{\theta,\phi}(t,t^{\prime})A S_{z}(t^{\prime},\theta)dtdt^{\prime}\\
    &=\underbrace{\int_{\Omega}\int_{\Omega}S_{z}(t,\theta)^{\intercal}A^{\intercal}C_{\phi}(t,t^{\prime})A S_{z}(t^{\prime},\theta)dtdt^{\prime}}_{H_{\theta}} - \\
    &\quad \int_{\Omega}\int_{\Omega}S_{z}(t,\theta)^{\intercal}A^{\intercal}[h_{\theta}^{\intercal}(t)H_{\theta}^{-1} h_{\theta}(t^{\prime})]A S_{z}(t^{\prime},\theta)dtdt^{\prime} \\
    &=H_{\theta} - \underbrace{[\int_{\Omega}S_{z}(t,\theta)^{\intercal}A^{\intercal}h_{\theta}^{\intercal}(t)dt]}_{H_{\theta}}H_{\theta}^{-1}\underbrace{[\int_{\Omega}h_{\theta}(t^{\prime})]A S_{z}(t^{\prime},\theta)dt^{\prime}]}_{H_{\theta}} \\
    &= H_{\theta} - H_{\theta}H_{\theta}^{-1}H_{\theta} = 0\in\mathbb{R}^{p\times p}
\end{aligned}
$$
Hence, the $p$-dimensional Gaussian random vector $\int f_{\theta}$ has zero mean and covariance, and thus equals zero with probability 1, as desired to show. 
\end{proof}

\section{Sensitivity Results}\label{sec:AppSens}
We calculated average sensitivity metrics by using the sensitivity systems defined for each model. We calculate a scalar index for the model sensitivity by first integrating the time series sensitivity to provide $\bar{S}_i = \int_\Omega S_i(t)dt$ for each dataset. We calculate the sensitivity at the data generating parameter. Then, we average over all 100 data instances to provide an average sensitivity, $\mathcal{S}_i=\frac{1}{100}\sum_{i=1}^{100}\bar{S}_i$. We then provide plots of a scaled sensitivity, $\mathcal{S}^*_i=\mathcal{S}_i/\underset{i}{\max}{\mathcal{S}_i}$ to visualize the relative impact of the parameters on each state variable.

\begin{figure}[ht!]
    \centering
    \includegraphics[width=0.7\linewidth]{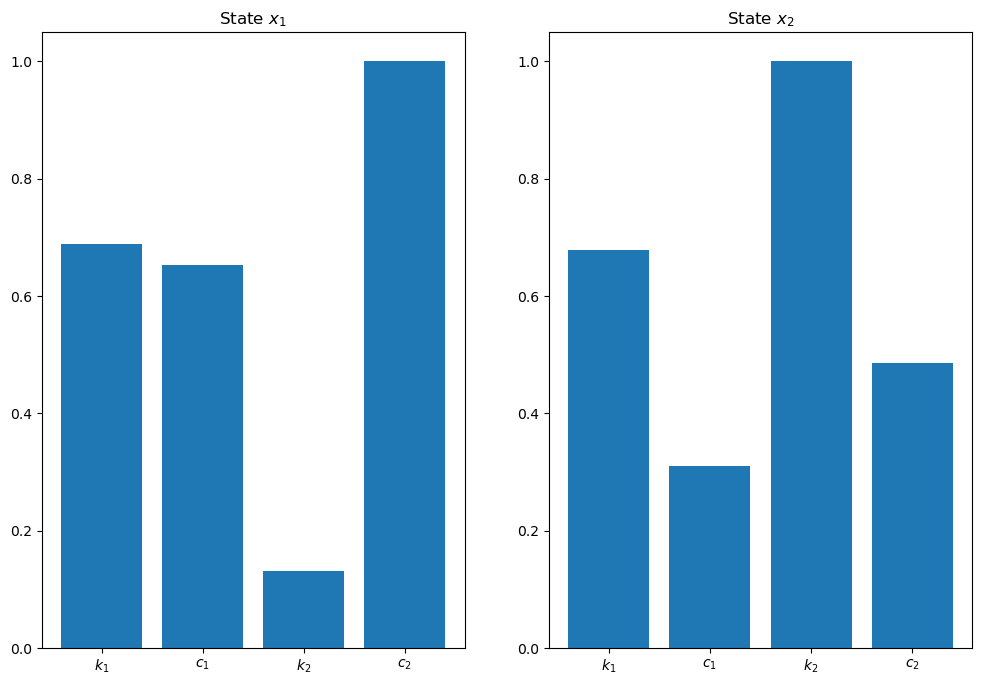}
    \caption{Sample averaged sensitivity metrics across the 100 data instances for the two-mass oscillator. Sensitivities are time-integrated for each state variable, then averaged across each dataset. Results are shown as sensitivities scaled to the maximum sensitivity in the system, $\mathcal{S}^*$.}
    \label{fig:TwoSpring_Sensitivity}
\end{figure}

\begin{figure}[ht!]
    \centering
    \includegraphics[width=\linewidth]{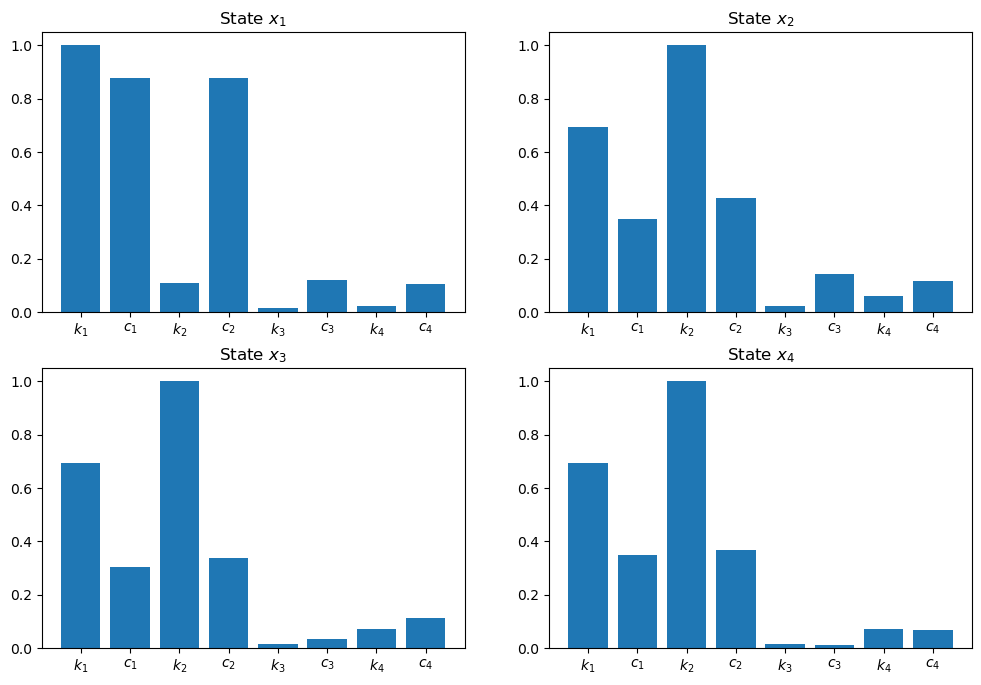}
    \caption{Sample averaged sensitivity metrics across the 100 data instances for the four-mass oscillator. Sensitivities are time-integrated for each state variable, then averaged across each dataset. Results are shown as sensitivities scaled to the maximum sensitivity in the system, $\mathcal{S}^*$.}
    \label{fig:FourSpring_Sensitivity}
\end{figure}

\begin{figure}[ht!]
    \centering
    \includegraphics[width=0.7\linewidth]{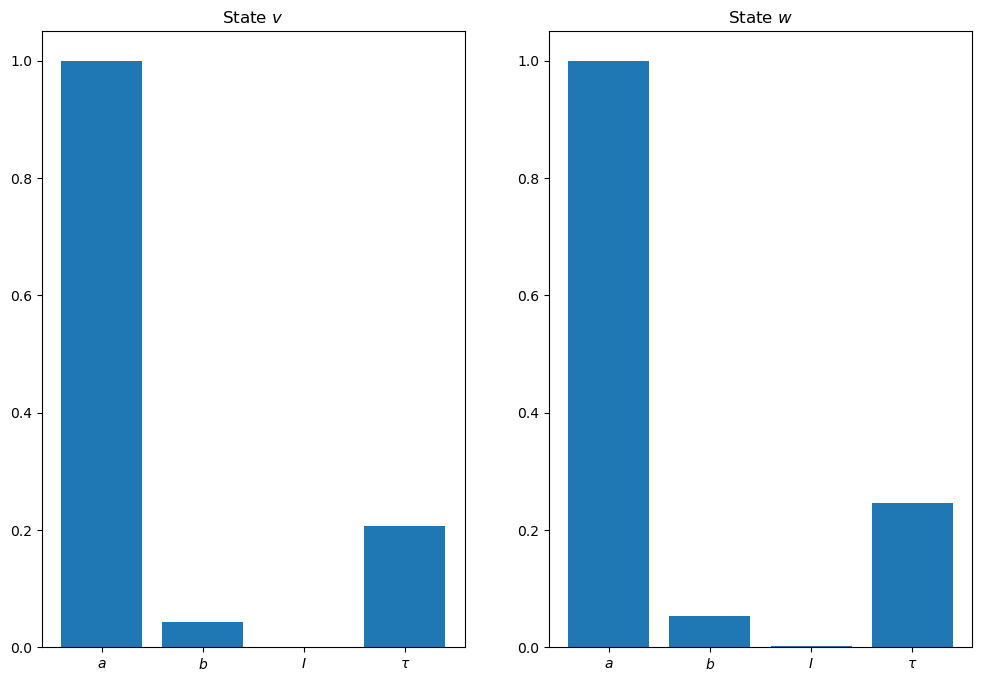}
    \caption{Sample averaged sensitivity metrics across the 100 data instances for the FitzHugh-Nagumo model. Sensitivities are time-integrated for each state variable, then averaged across each dataset. Results are shown as sensitivities scaled to the maximum sensitivity in the system, $\mathcal{S}^*$.}
    \label{fig:Fitz_Sensitivity}
\end{figure}

\begin{figure}[ht!]
    \centering
    \includegraphics[width=0.7\linewidth]{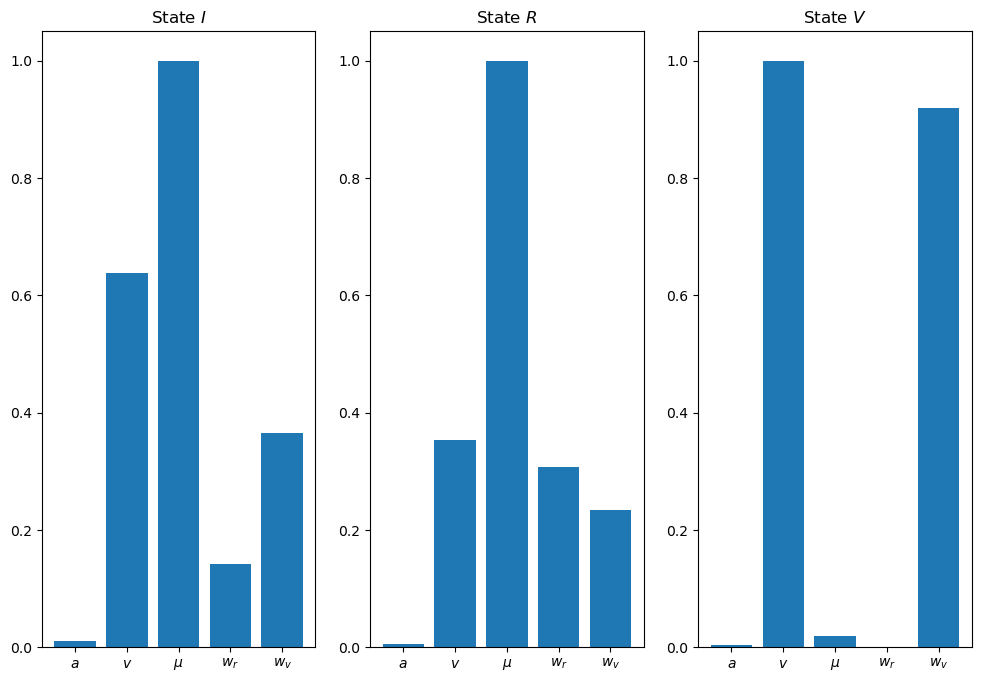}
    \caption{Sample averaged sensitivity metrics across the 100 data instances for the $SIRV$ model. We focus on the $I,R,$ and $V$ states. Sensitivities are time-integrated for each state variable, then averaged across each dataset. Results are shown as sensitivities scaled to the maximum sensitivity in the system, $\mathcal{S}^*$.}
    \label{fig:SIRV_Sensitivity}
\end{figure}

\end{document}